\documentclass[journal]{IEEEtran}

\usepackage{cite}

\usepackage[dvipsnames]{xcolor} 
\ifCLASSINFOpdf
   \usepackage[pdftex]{graphicx}
\else
\fi

\usepackage{amssymb} 
\usepackage{amsthm} 
\usepackage{amsmath}

\usepackage{algorithm} 
\usepackage{algpseudocode}

  \usepackage[caption=false,font=footnotesize]{subfig}

\usepackage{url}

\usepackage{defs}

\begin{document}
\title{{Extreme Compressed Sensing of Poisson Rates from Multiple Measurements}}

\author{Pavan~K.~Kota, 
        Daniel~LeJeune,~\IEEEmembership{Member,~IEEE}, Rebekah~A.~Drezek, 
        and~Richard~G.~Baraniuk,~\IEEEmembership{Fellow,~IEEE}
\thanks{P. K. Kota and R. A. Drezek are with the Department of Bioengineering, Rice University, Houston, TX 77005 USA. (e-mail: pkk1@rice.edu, drezek@rice.edu)}
\thanks{D. LeJeune and R. G. Baraniuk are with the Department of Electrical and Computer Engineering, Rice University, Houston, TX 77005 USA. (e-mail: dlejeune@rice.edu, richb@rice.edu)}
}

\IEEEpubid{ \begin{minipage}{\textwidth}\ \\[12pt]
~\copyright~2021 IEEE. Personal use of this material is permitted. Permission from IEEE must be obtained for all other uses in any current or future media, including reprinting/republishing this material for advertising or promotional purposes, creating new collective works, for resale or redistribution to servers or lists, or reuse of any copyrighted component of this work in other works.\end{minipage}}

\maketitle

\begin{abstract}
Compressed sensing (CS) is a signal processing technique that enables the efficient recovery of a sparse high-dimensional signal from low-dimensional measurements. In the multiple measurement vector (MMV) framework, a set of signals with the same support must be recovered from their corresponding measurements. Here, we present the first exploration of the MMV problem where signals are independently drawn from a sparse, multivariate Poisson distribution. We are primarily motivated by a suite of biosensing applications of microfluidics where analytes (such as whole cells or biomarkers) are captured in small volume partitions according to a Poisson distribution. We recover the sparse parameter vector of Poisson rates through maximum likelihood estimation with our novel Sparse Poisson Recovery (SPoRe) algorithm. SPoRe uses batch stochastic gradient ascent enabled by Monte Carlo approximations of otherwise intractable gradients. 
By uniquely leveraging the Poisson structure, SPoRe substantially outperforms a comprehensive set of existing and custom baseline CS algorithms. Notably, SPoRe can exhibit high performance even with one-dimensional measurements and high noise levels. This resource efficiency is not only unprecedented in the field of CS but is also particularly potent for applications in microfluidics in which the number of resolvable measurements per partition is often severely limited. We prove the identifiability property of the Poisson model under such lax conditions, analytically develop insights into system performance, and confirm these insights in simulated experiments. Our findings encourage a new approach to biosensing and are generalizable to other applications featuring spatial and temporal Poisson signals. 
\end{abstract}

\IEEEpeerreviewmaketitle

\maketitle

\section{Introduction}
As data increasingly informs critical decision-making, efficient signal acquisition frameworks must keep pace. Modern signals of interest are often high-dimensional but can be efficiently recovered by exploiting their underlying structure through signal processing.
The field of compressed sensing (CS), reviewed in~\cite{Baraniuk07-cslecnotes, Eldar12_csbook}, focuses on the recovery of sparse signals from fewer measurements than the signal dimension.
Concretely, an $N$-dimensional signal $\vx^*$ with at most $k$ nonzero entries (in which case $\vx^*$ is said to be $k$-sparse) can be recovered from a measurement vector $\vy$ acquired by $M$ sensors. The sensors' linear responses to entries of $\vx^*$ define a sensing matrix $\PhiB$ such that, compactly, $\vy=\PhiB \vx^*$. Recovering $\vx^*$ from $\vy$ is known as the single measurement vector (SMV) problem~\cite{Chen98-atomicdecomp, Donoho06-cs, Candes05-decodlinprog, Cohen09-csadaptive}. In the {\em multiple measurement vector} (MMV) problem~\cite{Chen05-MMV, Chen06-MMVtheory, Tropp06-MMVpt1somp}, $D$ measurements are captured in an $M \times D$ matrix $\YB$ in order to recover $\XB^*$, an $N \times D$ signal matrix. $\XB^*$ is jointly sparse such that only $k$ rows contain at least some nonzero elements. CS has been applied extensively in imaging~\cite{Duarte-singpixCS, Willett-imagingCSrev, Lustig-csMRI} and communications~\cite{Li13-iotCSrev, Zhu11-smvMUD, Ragheb08-CSadc} and only recently in biosensing~\cite{Cleary17, Koslicki14-WGSQuikr, Aghazadeh16-UMD, Ghosh20-covidCS}.

\begin{figure}
    \centering
    \includegraphics[trim=0 195 0 180, clip, width=1\linewidth]{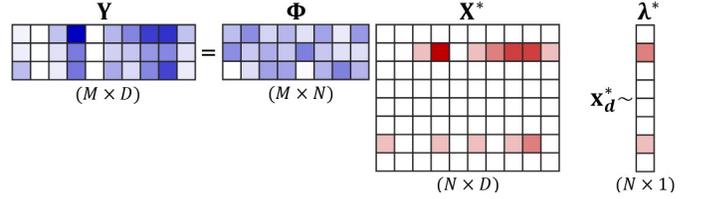}
    \caption{The {\em multiple measurement vector} (MMV) problem with Poisson signals (MMVP) with one sensor group and noiseless measurements. White squares are zeroes and darker colors represent larger values. Each column $\vx_d^*$ of $\XB^*$ is drawn from a Poisson distribution governed by the 2-sparse $\vlambda^ *$ (i.e., $\vx_d^* \overset{\mathrm{i.i.d.}}{\sim} \mathrm{Poisson}(\vlambda^*)$).}
    \label{fig:mmvp}
\end{figure}

Emerging \textit{microfluidics} technologies in the field of biosensing motivate a new MMV framework. With microfluidics, a single sample can be split into $D$ small-volume \textit{partitions} such as droplets or nanowells with $D$ on the order of $10^3$ to $10^7$~\cite{Guo12-dropthroughputrev}. 
Microfluidic partitioning captures individual analytes (e.g., cells~\cite{Hosokawa17-scseq, Linas13-scdroprev}; genes~\cite{Vogelstein99_dpcrOG}; proteins~\cite{Kim12_dpcr_molec, Yelles19-mobiledroplet}; etc.) in partitions, and analyte quantities across partitions are known to follow a Poisson distribution~\cite{Basu17-digitalrev, Moon11-poissondrop}. The common method to detect a library of analytes with large $N$ 
is to either dilute samples or split samples into more partitions
such that the Poisson distributions reduce to either empty or single-analyte capture, i.e., that columns of $\XB^*$ satisfy $\|\vx_d^*\| \in \{0,1\}$~\cite{Velez17-dHRMumd, Linas13-scdroprev}. This assumption motivates a straightforward $N$-class classification problem for each non-empty droplet, but it necessitates clear separation between classes even under noise, 
some prior knowledge of sample concentration, and the generation of many wasteful, empty partitions. We hypothesize that CS could generalize the signal recovery strategy when samples are sparse, a common characteristic of biological samples. For example, samples may contain only a few microbes or mutations of interest among many possibilities~\cite{Peters12-polymic, Vural15-sparsemutBC}. 

\IEEEpubidadjcol

We propose the following generally applicable framework for the MMV problem with Poisson signals (MMVP). Let each signal $\vx_d^*$ be drawn independently from a multivariate Poisson distribution parameterized by the $N$-dimensional, $k$-sparse vector $\vlambda^*$. That is, $x_{n,d}^* \sim \mathrm{Poisson}(\lambda_n^*)$ are independent. This framework should not be confused with the well-studied ``Poisson compressed sensing'' problem in imaging where the \emph{measurement noise}, rather than the signal, follows a Poisson distribution~\cite{Raginsky10-poissCS, Harmany11-spiraltap}. In contrast to typical MMV problems, our primary goal is to find an estimate $\widehat{\vlambda} \approx \vlambda^*$ from $D$ observations rather than to estimate $\widehat{\XB}$; however, given $\widehat{\vlambda}$, we will show that estimating $\widehat{\XB}$ is easily within reach if needed. Each signal and measurement pair ($\vx_d^*$, $\vy_d$) is in one of $G$ different \textit{sensor groups}, each with its own sensing matrix $\PhiB^{(g)}$ such that $\vy_d = \PhiB^{(g)}\vx^*_d$. 
The group $g$ associated with each index $d$ is known and deterministic.
In microfluidics, several sensor groups can be feasibly achieved by forking an input microfluidic channel into $G$ reaction zones each containing its own set of $M$ sensors. Note that $\vx_d^* \overset{\mathrm{i.i.d.}}{\sim} \mathrm{Poisson}(\vlambda^*)$ regardless of which group it is in. The statement $\YB = \PhiB \XB^*$ is the special case without noise where $G=1$ and is illustrated in Fig.~\ref{fig:mmvp}. For multiple groups with $\XB^{(g)*}$ denoting the submatrix of $\XB^*$ in group $g$, $\YB$ is the following concatenation:  
\begin{equation}
    \YB = \begin{bmatrix}
    \PhiB^{(1)}\XB^{(1)*} & ... & \PhiB^{(G)}\XB^{(G)*}
    \end{bmatrix}.
\end{equation}

\subsection{Contributions and Findings} We present the first exploration of the MMVP problem and develop a novel recovery algorithm and initial theoretical results. We take a maximum likelihood estimation (MLE) approach, treating $\YB$ as a set of $D$ observations from which to infer $\widehat{\vlambda}$. Our core contributions are
\textbf{1:}~the {\em Sparse Poisson Recovery} (SPoRe) algorithm that tractably estimates $\widehat{\vlambda}$ (Section~\ref{sec:spore});
\textbf{2:}~theoretical results on the identifiability of our MMVP model and insights into MLE performance (Section~\ref{sec:theo}); and
\textbf{3:}~simulations demonstrating SPoRe's superior performance over existing and custom baseline algorithms (Section~\ref{sec:sims}).
Although we are unable to provide theoretical guarantees for recovery, we analytically derive insights into the influence of various system parameters and confirm these insights in our simulated experiments.
We find that system designers should first maximize $M$ and then increase $G$ as necessary depending on the expected real-world conditions. While CS has always required $M>k$ for unique optima and stable recovery, the MMVP problem has no fundamental theoretical or practical lower bound for $M$. 

SPoRe's strong performance even with $M=1$ under very high measurement noise is unprecedented in CS and uniquely enables sensor-constrained applications in biosensing. Although microfluidics devices can rapidly generate a large number of partitions $D$ at a tunable rate, most optical and electrochemical sensing modalities that can keep pace are limited in $M$~\cite{Guo12-dropthroughputrev, Ngam19-droprev}. Commonly, fluorescently tagged sensors reveal droplets' contents rapidly as they flow by a detector, but spectral overlap generally limits $M$ to be five or less without highly specialized, system-specific approaches~\cite{Orth18_mxfluor}. High $M$ alternatives such as various spectroscopic techniques limit throughput, necessitate additional instrumentation, or complicate workflows \cite{Zhu13-dropsensingrev, Ngam19-droprev}. We speculate that these severe restrictions in $M$ may have forestalled research into CS's potential role in microfluidics.

\subsection{Previous Work}
To the best of our knowledge, the MMVP problem has not yet been explored, likely owing to the ongoing maturation of microfluidics and only recent application of CS to biosensing. The Poisson signal model constrains elements of $\XB$ to be nonnegative integers under a set of defined probability mass functions. Some aspects of this structure have been studied tangentially, but not the MMVP structure directly. 

The core MMV problem only imposes joint sparsity. Early greedy algorithms for this generalized scenario extend the classic Orthogonal Matching Pursuit (OMP) algorithm~\cite{Pati93-OMP} into OMPMMV~\cite{Chen06-MMVtheory}, simultaneously developed as Simultaneous OMP (S-OMP)~\cite{Tropp06-MMVpt1somp}. Generally, OMP-based algorithms iteratively build a support set of an estimated sparse solution $\widehat{\vx}$ (or $\widehat{\mathbf{X}}$) by testing for the correlation between columns of $\mathbf{\Phi}$ and the residuals between the measurements and previous estimates. A suite of greedy algorithms was recently developed that impose nonnegative constraints to a number of MMV approaches including OMP's analogues, and the nonnegative extensions outperformed their generalized counterparts~\cite{Kim16-MMVgreedyNN}.

The application of integer constraints to the SMV problem has proven challenging. Some theory involving sensing matrix design includes~\cite{Dai09-intCSwsc, Fukshansky19-intCSalg}, but practical algorithms have required additional constraints on the possible integers, e.g., $\vx \in \{0,1\}^N$ or other finite-alphabet scenarios~\cite{Nakarmi12-binaryCS, Tian09-intCSfinalph, Draper09-finfieldCS, Shim14_sMMP, Zhu11-smvMUD}. A recent study verified that these problems, as well as those with unbounded integer signals, are NP-hard~\cite{Lange20-csIntNP}. Algorithms for the unconstrained integer SMV problem thus apply greedy heuristics such as OMP-based approaches~\cite{Flinth17-PROMP, Sparrer16-intCSmmse}. 

Additional structural constraints can also make these problems tractable. The communications problem of multi-user detection (MUD), reviewed in~\cite{Alam-MUDrev}, bears some similarity to MMVP. Here, the activity of $N$ users is the signal of interest and generally follows a Bernoulli model where each user is active with the same prior probability $p_a$~\cite{Zhu11-smvMUD}. An alternative prior with $\sum_{n=1}^N x_{n,d} \sim \mathrm{Poisson}(\lambda)$ models the mean number of total active users in any given signal~\cite{Ji-MUDpoisson} although the authors solely explored an overdetermined system. Applying an MMV framework to MUD enables underdetermined ($M<N$) applications but has generally assumed that any active user is active for the entire frame of observation (a row of $\XB$ is entirely zero or nonzero)~\cite{Monsees-mmvMUDmap, Liu-mmvMUDiot}. Recently, the potency of sensor groups with a $G=2$ system was demonstrated in the MUD context~\cite{Alam-MUDmeasgroups}. Despite some similarities to MMVP with an MMV framework and discrete signals, MUD most fundamentally differs from MMVP in its utilization of the probabilistic structure of $\XB$. In MUD, the model parameters governing user activities are assumed and leveraged in recovery of $\XB$, whereas in MMVP, the model parameters in $\vlambda^*$ themselves are the target of recovery.

\section{Sparse Poisson Recovery (SPoRe) Algorithm}
\label{sec:spore}

\subsection{Notation}

We denote by $P(\cdot)$ a probability mass function and by $p(\cdot)$ a probability density function. We use $\mathbb{R}^N$ and $\mathbb{Z}^N$ to represent the $N$-dimensional Euclidean space and integer lattice, respectively. We denote by $\mathbb{R}_+^N$ and $\mathbb{Z}_+^N$ the non-negative restrictions on these spaces. We use script letters ($\mathcal{A}, \mathcal{B}$, ...) for sets unless otherwise described. We use lowercase and uppercase bold-face letters for vectors and matrices, respectively. We represent their dimensions with uppercase letters (e.g., $\XB \in \mathbb{Z}_+^{N \times D}$) that are indexed by their lowercase counterparts.
For example, $x_{n,d}$ is the element of $\XB$ in the $n$th row and $d$th column, and we use the shorthands $\vx_n$ and $\vx_d$ to represent the entire $n$th row and $d$th column vectors, respectively. Other lower case letters ($a$, $b$, $\epsilon$, etc.) may represent variable or constant scalars depending on context. We use $\vlambda^*$ and $\XB^*$ to refer to the true signal values, and we denote estimates $\widehat{\vlambda}$ and $\widehat{\XB}$ with the source of the estimate (e.g., MLE, SPoRe, baseline algorithm) being implicit from the context. We denote the null space of matrix $\mA$ by $\mathcal{N}(\mA)$. As one abuse of notation, for densities of the form $p(\vy|\vx)$, we let the corresponding $\PhiB^{(g)}$ applied to $\vx$ and the relevant noise model be implicit. Also, we let the division of two vectors represent element-wise division.

\subsection{Algorithm}

If the index $d$ is in sensor group $g$, we say that linear measurements are corrupted by an additive random noise vector $\vb_d$: 
\begin{equation} \label{eq:meas_model}
    \vy_d = \PhiB^{(g)} \vx_d + \vb_d.
\end{equation}
We let $\vb_d$ be entirely independent (e.g., additive white Gaussian noise (AWGN), as used in our simulations) or dependent on $\vx$. 
With $\vx_d \overset{\mathrm{i.i.d.}}{\sim} \mathrm{Poisson}(\vlambda^*)$, $\vy_d$ are independent across $d$ as well. The MLE estimate maximizes the average log-likelihood of the measurements: 
\begin{align}
    \widehat{\vlambda}_{MLE} &= \argmax_{\vlambda} \prod_{d=1}^D p(\vy_d|\vlambda) \\
    &= \argmax_{\vlambda} \frac{1}{D} \sum_{d=1}^{D} \log \sum_{\vx \in \mathbb{Z}_+^N} p(\vy_d|\vx)P(\vx|\vlambda). \label{eq:ll_logsum}
\end{align}
Our Sparse Poisson Recovery (SPoRe) algorithm (Algorithm~\ref{algo_spore}) optimizes this function with batch stochastic gradient ascent, drawing $B$ elements uniformly with replacement from $\{1, ..., D \}$ to populate a batch set $\mathcal{B}$. First, note that
\begin{equation}
    \nabla_{\vlambda} P(\vx|\vlambda) = P(\vx|\vlambda) \bigg(\frac{\vx}{\vlambda}-1\bigg).
\end{equation}
Denoting the objective function from the right-hand side of~$\eqref{eq:ll_logsum}$ as $\ell$, the gradient is
\begin{equation}
    \nabla_{\vlambda} \ell = \frac{1}{B} \sum_{d\in \mathcal{B}} \frac{\sum_{\vx \in \mathbb{Z}_+^N}p(\vy_d|\vx) P(\vx|\vlambda) \vx}{ \vlambda \sum_{\vx \in \mathbb{Z}_+^N} p(\vy_d|\vx)P(\vx|\vlambda) } - 1. \label{eq:grad_general}
\end{equation}
With gradient ascent, each iteration updates $\vlambda \gets \vlambda + \alpha \nabla_{\vlambda} \ell$ with learning rate $\alpha$. However, the summations over all of $\mathbb{Z}_+^N$ are clearly intractable. SPoRe approximates these quantities with a Monte Carlo (MC) integration over $S$ samples of $\vx$, newly drawn for each batch gradient step from sampling distribution $Q: \mathbb{Z}_+^N \rightarrow \mathbb{R_+}$, such that
\begin{align}
\sum_{\vx\in\mathbb{Z}_+^N} p(\vy | \vx) P(\vx | \vlambda) \approx \frac{1}{S} \sum_{s=1}^S \frac{p(\vy | \vx_s) P(\vx_s | \vlambda)}{Q(\vx_s)}.
\end{align}
The optimal choice of $Q(\vx_s)$ is beyond the scope of this work, but we found that $Q(\vx_s) = P(\vx_s|\vlambda)$ simplifies the expression, is effective in practice, and draws inspiration from the popular expectation--maximization algorithm \cite{Dempster77-EM}. In other words, the sampling function is updated at each iteration based on the current estimate of $\vlambda$. The gradient thus simplifies to
\begin{equation} \label{eq:grad_final}
    \nabla_{\vlambda} \ell = \frac{1}{B} \sum_{d \in \mathcal{B}} \frac{\sum_{s=1}^{S}p(\vy_d|\vx_s)\vx_s}{\vlambda \sum_{s=1}^{S}p(\vy_d|\vx_s)}  - 1.
\end{equation}
Note that if only one $\widehat{\vx}_d \in \mathbb{Z}_+^N$ satisfied $p(\vy_cd|\widehat{\vx}_d) > 0$ for every $\vy_d$, the objective $\ell$ would be concave with $\widehat{\vlambda} = \frac{1}{D} \sum_{d=1}^D \widehat{\vx}_d$, i.e., the MLE solution if $\XB^*$ were directly observed. Of course, with compressed measurements and noise, multiple signals may vie to ``explain'' any single measurement, but 
SPoRe's key strength is that it jointly considers independent measurements to directly estimate $\widehat{\vlambda}$. 

We note that for finite samples, since the MC integration occurs inside a logarithm, the stochastic gradient is biased. However, since it converges in probability to the true gradient, we can expect results comparable to SGD with an unbiased gradient for sufficiently large $S$~\cite{Chen18-consistentSGD}. 

\begin{figure}
  \centering
  \includegraphics[trim=0 50 0 50, clip, width=\fw \linewidth]{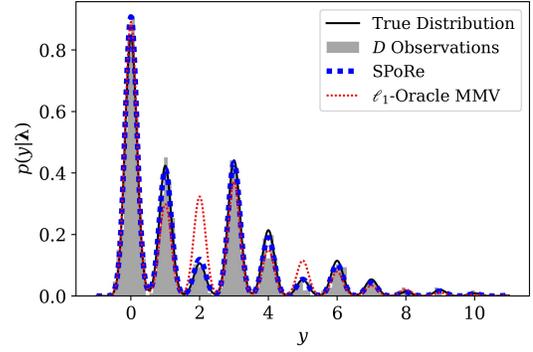}
  \caption{Example of MMVP and Sparse Poisson Recovery (SPoRe) with $M<k<N$: $\PhiB = [1, 2, 3]$, $\vlambda^* = [0.5, 0, 0.5]$, and $D=1000$ measurements under additive white Gaussian noise (AWGN) $\vb \sim \mathcal{N}(0, \sigma^2)$ with $\sigma^2 = 0.02$. SPoRe attempts to fit the distribution of measurements directly and finds $\widehat{\vlambda} \approx [0.45, 0.03, 0.44]$. For comparison, the $\ell_1$-Oracle (see Section~\ref{sec:sims}) minimizes the measurement error $\widehat{\XB} = \arg \min_{\XB} \| \YB-\PhiB \XB\|_F$ with $\XB\geq 0$ and $\sum_{n,d} x_{n,d} = \sum_{n,d} x_{n,d}^*$ as affine constraints. The estimate $\widehat{\vlambda}$ for the $\ell_1$-Oracle is then set to the average of the columns of $\widehat{\XB}$, and in this example,
  $\widehat{\vlambda} \approx [0.33, 0.31, 0.31]$. The distributions $p(\vy|\widehat{\vlambda})$ for each estimation method are compared against the true distribution $p(\vy|\vlambda^*)$ and the empirical histogram of the $D$ observations.} 
\label{fig:gmm}
\end{figure}

Fig.~\ref{fig:gmm} illustrates key concepts of SPoRe and MMVP with a small example where $M=1$ and $\vlambda^* = [0.5, 0, 0.5]$ for which we can numerically compute $p(\vy|\vlambda)$ for various $\vlambda$. The measurements $\vy_d$ are effectively drawn from an underlying mixture distribution depending on the noise; e.g., under AWGN, $\vy_d$ follows a Gaussian mixture. The weights on each mixture component are controlled by $\vlambda$. In simulated recovery, SPoRe assigns weights to the mixture via $\widehat{\vlambda}$ according to the distribution of measurements, coming close to the true underlying distribution. In contrast, an $\ell_1$-Oracle (Section~\ref{sec:sims}) which represents best-case performance for a standard, convex sparse recovery process fails because $M<k$ as shown by its error in $\vlambda$ and illustrated by the difference in the distributions. 
Moreover, by using $\PhiB = [1,2,3]$, many $\vx$ will map to the same $\vy$. While CS theory generally focuses on conditions for unique or well-spaced projections of $k$-sparse signals (e.g., the restricted isometry property, RIP~\cite{Candes05-decodlinprog}), we demonstrate that such restrictions are unnecessary in MMVP.
By accounting for the latent Poisson distribution in the signals, SPoRe succeeds even when $M<k$.

Algorithm~\ref{algo_spore} summarizes the implementation details of SPoRe. Even though $\vlambda \in \mathbb{R}_+^N$, we enforce $\vlambda \geq \epsilon$ by clipping ($\epsilon = 10^{-3}$ in our simulations) to maintain exploration of the parameter space. Note that in \eqref{eq:grad_final}, gradients can become very large with finite sampling as some elements of $\vlambda$ approach zero.
We found that rescaling gradients to maximum norm $\gamma$ helps stabilize convergence. For rescaling, we consider only the subvector $\vdelta_\Gamma$ of the $\alpha$-scaled gradient $\vdelta$, defining indices $n \in \Gamma \subseteq \{1, ... , N\}$ if $\lambda_n + \delta_n > \epsilon$. This restriction ensures that rescaling is solely based on the indices still being optimized, excluding those clipping to $\epsilon$. 

\begin{algorithm}
\caption{Sparse Poisson Recovery (SPoRe) }\label{algo_spore}
\begin{algorithmic}[1]
\Statex{\textbf{Input}: $\vlambda^{(0)}$, $B$, $S$, $\gamma$, $\alpha$, $\epsilon$}

\State $\vlambda \gets \vlambda^{(0)}$
\State $i=0$
\Repeat
    \State Draw $B$ columns of $\YB$ uniformly with replacement
    \State Draw $S$ new samples from $Q(\vx_s)$
    \State $\vdelta \gets \alpha \nabla_{\vlambda} \ell(\vlambda)$ \Comment{\eqref{eq:grad_final}}

    \If{$\|\vdelta_\Gamma \|_2 > 
    \gamma$}
        \State{$\vdelta \gets \frac{\gamma}{\|\vdelta_\Gamma \|_2} \vdelta$} \Comment{Rescale gradient step}
    \EndIf

    \State $\lambda_n \gets \max (\lambda_n+\delta_n, \epsilon)$
\Until{stopping criterion met}

\State \textbf{return} $\vlambda$
\end{algorithmic}
\end{algorithm}

For our stopping criterion, we evaluate a moving average of $\widehat{\vlambda}$ for convergence. We also track the estimated value of the objective function $\ell(\vlambda)$, reduce $\alpha$ if no improvements in $\ell(\vlambda)$ have been seen within a patience window, and terminate if $\alpha$ is reduced three times. We conducted all experiments on commodity personal computing hardware. Ultimately, recovery of $\widehat{\vlambda}$ takes a few minutes on a single core, and SPoRe can be easily parallelized in the future for faster performance.

\subsection{Practical Considerations}
\label{subsec:spore_practical}
Within an iteration, we found that using the same $S=1000$ samples for all $d\in \mathcal{B}$ helped to vectorize our implementation to dramatically improved speed over sampling $S$ times for \emph{each} drawn $\vy_d$. This simplification had no noticeable influence on performance. While we found random initializations with a small offset $\vlambda^{(0)} \sim \mathrm{Uniform}(0,1) + \nu$ (with $\nu = 0.1$) to be effective in general, we encountered a numerical issue when under low-variance AWGN. Even though AWGN results in nonzero probabilities everywhere, $p(\vy_d|\vx_s)$ may numerically evaluate to zero for all drawn samples in low-noise settings. These zeros across all samples result in undefined terms in the summation over $d \in \mathcal{B}$ in \eqref{eq:grad_final}. SPoRe simply ignores such undefined terms, but when this numerical issue occurs for \emph{all} of $\mathcal{B}$, SPoRe takes no gradient step. With very low noise and large $N$ dampening the effectiveness of random sampling, SPoRe may stop prematurely as it appears to have converged. This problem did not arise with larger noise variances where even inexact samples pushed $\widehat{\vlambda}$ in the generally appropriate direction until better samples could be drawn (recall that $Q(\vx) = P(\vx|\widehat{\vlambda})$ at each iteration). Nonetheless, we decided to set $\vlambda^{(0)} = \nu$ for consistency across all simulated experiments. We speculate that setting $\vlambda^{(0)}$ to a small value helped encourage sampling sparse $\vx$'s in early iterations to help find $\vx_s$ with nonzero $p(\vy_d|\vx_s)$, bypassing the numerical issue altogether.

\newtheorem{theorem}{Theorem}[section]
\newtheorem{definition}[theorem]{Definition}
\newtheorem{corollary}[theorem]{Corollary}
\newtheorem{lemma}[theorem]{Lemma}
\newtheorem{proposition}[theorem]{Proposition}

\section{Theory and Analysis} \label{sec:theo}
The summation over $\vx$ inside the logarithm of the objective function complicates the precise analysis of SPoRe.
However, we can consider the asymptotic MMVP problem as $D \to \infty$ and its MLE solution to gain insight into the superior recovery performance of SPoRe and to understand design trade-offs. In this section, we prove the sufficiency of a simple null space condition on $\PhiB$ for identifiability of our MLE model. 
We then characterize the loss in Fisher Information for MMVP and show how losses accrue with the increase of signals that map to the same measurements. 
Lastly, we derive insights into the influence of sensor groups through a small-scale analysis. From a system design standpoint, we find that designers should first increase $M$ as much as feasible and then increase $G$ as needed. 
All proofs can be found in the Appendix.

\subsection{Identifiability of MMVP Models} \label{subsec:ident}

Identifiability refers to the uniqueness of the model parameters that can give rise to a distribution of observations. A model $\mathcal{P} = \{p(\cdot | \vlambda) : \vlambda \in \mathbb{R}_+^{N}\}$ is a collection of distribution functions which are indexed by the parameter $\vlambda$; in the MMVP problem, each choice of $\PhiB$ and noise give rise to a different model $\mathcal{P}$. Through an optimization lens, if our model is identifiable, then $\vlambda^*$ is the unique global optimum of the data likelihood as $D \rightarrow \infty$. Recall that $p(\vy|\vlambda) = \sum_{x\in\mathbb{Z}_+^N} p(\vy|\vx)P(\vx|\vlambda)$, meaning that we can interpret this model as each sensor group consisting of a mixture whose elements' positions are governed by $\PhiB^{(g)} \vx$, distributions by the noise model, and weights by $P(\vx|\vlambda)$. We focus in this analysis on a single sensor group, since as $D\rightarrow \infty$, at least one sensor group contains infinite measurements. If the corresponding $\PhiB^{(g)}$ satisfies the conditions we describe here, then the model is identifiable. Formally:

\begin{definition} [Identifiability]
The model $\mathcal{P}$ is identifiable if $p(\vy|\vlambda) = p(\vy|\vlambda') \ \forall \vy \Rightarrow \vlambda = \vlambda'$ for all $\vlambda, \vlambda' \in \mathbb{R}_+^N$. 
\end{definition}

The identifiability of mixtures is well-studied \cite{Teicher63-IDfinmix, Tallis69-IDinfmat}; if a mixture is identifiable, the mixture weights uniquely parameterize possible distributions. For finite mixtures, a broad set of distributions including multivariate exponential and Gaussian have been proven to be identifiable \cite{Yakowitz68-IDmultivarFinite}. A finite case may manifest in realistic MMVP systems where measurements $\vy$ must eventually saturate; all sensors have a finite dynamic range of values they can capture. In the most general case, $p(\cdot|\vlambda)$ is a countably infinite mixture. Although less studied, countably infinite mixtures are identifiable under some classes of distributions \cite{Yang13-IDcountInf}. The AWGN that we use in our simulations is identifiable for both the finite and countably infinite cases. Characterizing the full family of noise models that are identifiable under countably infinite mixtures is beyond the scope of this work. 
Our contribution is that \emph{given} a noise model that yields identifiable mixtures, equal mixture weights induced by $\vlambda$ and $\vlambda'$ imply $\vlambda=\vlambda'$. We prove the sufficiency of the following simple conditions on $\PhiB$ for identifiability:

\begin{theorem} [Identifiability of Mixture Weights] \label{th:id_weights}
Let $\vb$ be additive noise drawn from a distribution for which a countably infinite mixture is identifiable. If~$\mathcal{N}(\PhiB)\cap\mathbb{R}_+^N=\{\mathbf{0}\}$ and $\vphi_n \neq \vphi_{n'} \ \forall n, n' \in \{1, \ldots , N\}$ with $n \neq n'$, then $\mathcal{P}$ is identifiable.  
\end{theorem}

The null space condition essentially says that any nonzero vector in $\mathcal{N}(\mathbf{\Phi})$ must contain both positive and negative elements. Many practical $\mathbf{\Phi}$ satisfy this constraint (e.g., any $\mathbf{\Phi}$ with at least one strictly negative or positive row). The second condition is trivial: no two columns of $\mathbf{\Phi}$ can be identical. We also obtain a separate sufficient condition, that $\PhiB$ drawn from any continuous distribution results in identifiability. 

\begin{corollary} [Identifiability with Random Continuous $\PhiB$] \label{cor:id_continuous}
Let $\vb$ be additive noise drawn from a distribution for which a countably infinite mixture is identifiable. If the elements of $\PhiB$ are independently drawn from any continuous distribution, then $\mathcal{P}$ is identifiable.
\end{corollary}

We emphasize the general result of Theorem~\ref{th:id_weights}, since discrete sensing is common in biomedical systems. For example, sensors are often designed to bind to an integer number of known target sites and yield ``digital'' measurements~\cite{Basu17-digitalrev,Zhang12_specDNA}. 
Discrete $\PhiB$ can give rise to what we call \textit{collisions}. Formally:

\begin{definition} [Collisions and Collision Sets] \label{defn_coll}
Let $\PhiB \in \mathbb{R}^{M \times N}$ be a sensing matrix applied to signals $\vx \in \mathbb{Z}_+^N$. A \textit{collision} occurs between $\vx$ and $\vx'$ when $\PhiB \vx = \PhiB \vx'$. A \textit{collision set} for an arbitrary $\vu \in \mathbb{Z}_+^N$ is the set $\mathcal{C}_{\vu} = \{ \vx: \PhiB \vx = \PhiB \vu; \vx \in \mathbb{Z}_+^N \}$. 
\end{definition}

If the distribution from which $\vb$ is drawn is fixed (e.g., AWGN) or a function of $\PhiB \vx$, then the mixture weights are the probability mass of each collision set. Let the set of collision sets be $\mathcal{U}$ with $\mathcal{C}_\vu \in \mathcal{U}$ being an arbitrary collision set.

\begin{gather} 
    p(\vy|\vlambda) = \sum_{\mathcal{C}_\vu \in \mathcal{U}} p(\vy|\vx \in \mathcal{C}_\vu) P(\mathcal{C}_\vu|\vlambda) \label{eq:mix_coll} \\
    P(\mathcal{C}_\vu|\vlambda) = \sum_{\vx \in \mathcal{C}_\vu} P(\vx|\vlambda).
\end{gather}
The weights of the mixture elements are governed by $P(\mathcal{C}_\vu|\vlambda)$. Given a noise model that yields identifiable mixtures, the same distribution of observations $\vy$ implies that the mixture weights are identical, i.e. $P(\mathcal{C}_\vu|\vlambda) = P(\mathcal{C}_\vu|\vlambda') \ \forall \vu$. We prove that $P(\mathcal{C}_\vu|\vlambda) = P(\mathcal{C}_\vu|\vlambda') \ \forall \vu$ implies $\vlambda = \vlambda'$, which implies the identifiability of $\mathcal{P}$ under both the conditions of Theorem \ref{th:id_weights} and Corollary \ref{cor:id_continuous}.

Our proofs are based on the existence and implications of single-vector collision sets $\mathcal{C}_{\vx} = \{\vx\}$. When~\eqref{eq:mix_coll} holds, $\vu$ indexes both the mixture elements and the collision sets. In the general case where $\vb$ is dependent on $\vx$ and not simply $\PhiB \vx$, signals participating in the same mixture element may have different noise distributions. These differences can only further subdivide collision sets and leaves single-vector collision sets unaffected. Thus, our results also cover the general noise case.

\subsection{Fisher Information of MMVP Measurements}\label{subsec:fisher}

While identifiability confirms that $\vlambda^*$ is a unique global optimum of the MLE problem given infinite observations, Fisher Information helps characterize estimation of $\vlambda^*$ as $D$ increases. The Fisher Information matrix $\mathcal{I}$ is the (negative) Hessian of the expected log-likelihood function at the optimum $\vlambda^*$, and it is well-known that under a few technical conditions the MLE solution is asymptotically Gaussian with covariance 
$\mathcal{I}^{-1}/D$. Intuitively, higher Fisher Information implies a ``sharper'' optimum that needs fewer observations for stable recovery. For direct observations of Poisson signals $\vx^*_d$ rather than $\vy_d$, $\mathcal{I}$ is diagonal with $\mathcal{I}_{n,n} = 1/\lambda^*_n$.
In MMVP with observations of noisy projections ($\vy_d$), $\mathcal{I}$ and its inverse are difficult to analyze. We can, however, instead characterize the reduction in $\mathcal{I}_{n,n}$ in MMVP caused by the noisy measurement of $\vx_d^*$ and derive an insight that we empirically confirm in Section~\ref{subsec:efficiency}. Concretely, elements of $\mathcal{I}$ follow
\begin{equation} \label{fi_core}
    \mathcal{I}_{i,j} = \mathbb{E} \bigg[\bigg (\frac{\partial}{\partial \lambda^*_i} \log p(\vy|\vlambda^*)\bigg) \bigg( \frac{\partial}{\partial \lambda^*_j} \log p(\vy|\vlambda^*)\bigg) \bigg ].
\end{equation} 
We denote the shorthand $w_\vx \defeq p(\vy|\vx)P(\vx|\vlambda)$ and note that $\sum_{\vx} w_\vx = p(\vy|\vlambda)$. Following a similar derivation for the partial derivatives in~\eqref{eq:grad_final}, it can be shown that the general expression for diagonal elements $\mathcal{I}_{n,n}$ is
\begin{equation}
    \mathcal{I}_{n,n} = \int \bigg(\frac{\sum_\vx w_\vx{x_n}}{\sum_\vx w_\vx\lambda^*_n}  - 1\bigg)^2 \bigg(\sum_\vx w_\vx \bigg) d \vy. \label{eq:I_gen}
\end{equation}
In the ideal scenario, we observe $\vx^*_d$ directly such that
\begin{gather}
    \mathcal{I}_{n,n}^\rmideal = \sum_\vx P(\vx|\vlambda) \bigg(\frac{x_n}{\lambda_n} - 1\bigg)^2  \bigg(\int p(\vy|\vx) d \vy\bigg) . \label{eq:I_ideal}
\end{gather}
It can be easily shown that~\eqref{eq:I_ideal} reduces to the canonical $1/\lambda^*_n$. The integration of $p(\vy|\vx)$ evaluates to one, but we can manipulate it algebraically to re-express the quantity as
\begin{gather}
    \mathcal{I}_{n,n}^\rmideal = \int \sum_\vx \bigg[ w_\vx \bigg(\frac{x_n}{\lambda_n} - 1\bigg)^2 \bigg] d \vy. \label{eq:I_ideal_reex}
\end{gather}
Let $\mathcal{I}_{n}^\rmloss \defeq \mathcal{I}_{n,n}^\rmideal - \mathcal{I}_{n,n}$ and let $\sum_{(\vx',\vchi)}$ denote 
the sum over all pairs of signals $\vx', \vchi \in \mathbb{Z}_+^N$. Expanding Equations~\eqref{eq:I_gen} and~\eqref{eq:I_ideal_reex} and simplifying yields
\begin{align}
    \mathcal{I}_{n}^\rmloss
    &= \frac{1}{{\lambda^*_n}^2} \int \bigg(\sum_\vx w_\vx x_n^2 - \frac{(\sum_\vx w_\vx x_n)^2}{\sum_\vx w_\vx} \bigg) d \vy \nonumber \\ 
    &= \frac{1}{{\lambda^*_n}^2} \int \frac{1}{\sum_\vx w_\vx} \bigg( \sum_{\forall \vx',\vchi} w_{\vx'} w_{\vchi} (x'_n - \chi_n)^2 \bigg) d \vy . \label{eq:fi_loss}
\end{align}

Note that $\mathcal{I}_{n}^\rmloss$ is non-negative such that $\mathcal{I}_{n,n} \leq \mathcal{I}_{n,n}^\rmideal$ and that pairs of signals with $x'_n \neq \chi_n$ can contribute to $\mathcal{I}_{n}^\rmloss$. Also note that, $w_{\vx'} w_{\vchi} = p(\vy|\vx') p(\vy|\vchi) P(\vx'|\vlambda^*)P(\vchi|\vlambda^*)$ and that $P(\vx|\vlambda^*)>0$ only when $\mathrm{supp}(\vx) \subseteq \mathrm{supp}(\vlambda^*)$. Thus, the Fisher Information is only reduced over the direct Poisson observation case when there are pairs of signals that are well-explained by the same $\vy$ and also likely Poisson signals. Clearly, $\vlambda^*$ with higher $k$ will result in more of such pairs, which we confirm in Section~\ref{subsec:efficiency}. Although further precise analysis via Fisher Information is challenging, we provide deeper analysis of the special case of the MMVP problem with small $\vlambda$ through a different lens in the next section.

\subsection{Small Scale Analysis} \label{subsec:small}

With identifiability, we know that $\vlam^*$ uniquely maximizes the expected log-likelihood. However, because SPoRe uses stochastic gradient ascent to optimize the empirical log-likelihood, it will typically achieve a $\widehat\vlam$ that is near but not equal to $\vlam^*$. We therefore wish to understand how the neighborhood of $\vlam^*$ changes given the parameters of the problem. The natural way to do this for MLE problems is to consider the Fisher Information matrix as in the previous section, but the presence of a sum inside the logarithm makes analysis difficult. Instead, we consider a particular $\widetilde{\vlam}$ near $\vlam^*$ that solves an optimization related to the original likelihood maximization problem. To further simplify the setting, we consider the ``small scale'' case where $\sum_n \lambda_n^*$ is small enough that there is almost never a case where $\sum_n x_n^* > 1$. We emphasize that although this setting is simple, the MLE approach can still drastically outperform a trivial solution such as $\widehat\vlam = \E[\widehat \vx]$, where $\widehat \vx = \argmax_\vx p(\vy | \vx)$, since with sufficient noise, $\widehat \vx \neq \vx^*$ with arbitrary probability (Section~\ref{subsec_sims_int}).

At the small scale, the distribution of each $x_n^*$ becomes Bernoulli with parameter $\lambda_n$, and the probability that $x_n^* = 1$ and $x_{n'}^* = 1$ for $n \neq n'$ vanishes. Let $n^* \defeq$ (the first nonzero index of $\vx^*$, 0 if none), which has a categorical distribution with parameter $\vlam^*$. We abuse notation so that $\vphi_0 = \vzero$, $\lambda_0^*$ is the probability that $n^* = 0$, and $\sum_{n=0}^N \lambda_n^* = 1$. Applying Jensen's inequality to the log-likelihood for the conditional expectation given $n^*$, we obtain 
\begin{align}
    \E \left[ \log \sum_{n=0}^N p(\vy | n) \lambda_n \right]
    \leq
    \E_{n^*} \left[ \log \sum_{n=0}^N \mathbb{E}_{\vy|n^*} \left[p(\vy | n) \right] \lambda_n \right].
\end{align}
Call the right-hand side of this inequality the \emph{Jensen bound}.
This Jensen bound via the logarithm has the attractive property of having a gradient that is equal to a first-order Taylor approximation of the gradient of the original likelihood.\footnote{The Taylor expansion is of $f(u, v) = u/v$, for which a first-order approximation yields $\E[U/V] \approx \E[U]/\E[V]$ for random variables $U$, $V$.} To see this, consider the partial derivatives for a single $\lambda_n$:
\begin{align}
    \E \left[ \frac{p(\vy | n)}{\sum_{n'=0}^N p(\vy | n') \lambda_{n'}} \right]
    \approx
    \E_{n^*} \left[ \frac{ \mathbb{E}_{\vy | n^*} [p(\vy | n)] 
    }{
    \sum_{n'=0}^N \mathbb{E}_{\vy | n^*} [p(\vy | n')] \lambda_{n'}
    }
    \right].
\end{align}
Thus, we can expect the optimizer of the Jensen bound to be close to $\vlambda^*$ (this is particularly true as measurement noise vanishes and the bound becomes tight).

In the case where $G=1$ under AWGN, we have the following result characterizing the~solution~of~the~Jensen~bound.
\begin{proposition}
    \label{prop:small-single}
    If $\vy \sim \mathcal{N}(\vphi_{n^*}, \sigma^2 \mI)$ and
    \begin{align}
        \mK = \left( \exp \left\{ -\frac{1}{4 \sigma^2} \|\vphi_n - \vphi_{n'} \|_2^2 \right\} \right)_{n, n'=0}^N
    \end{align}
    is invertible, then the maximizer $\widetilde{\vlam}$ of the Jensen bound satisfies
    \begin{align}
        \widetilde{\vlam} \propto \mK^{-1} \left( \frac{\vlambda^*}{\mK^{-1} (\vs - \vmu)} \right).
    \end{align}
    where $\vs \in \partial \|\widetilde{\vlam}\|_1$ and for all $n$, $\mu_n \geq 0$ and $\mu_n \tilde \lambda_n = 0$. 
\end{proposition}

In the case where all entries of $\widetilde{\vlam}$ are positive, $\vs - \vmu = \vone$. $\mK$ has values of one along the diagonal and smaller values off the diagonal, so it mimics the identity matrix. Clearly, as $\mK \to \mI$, $\widetilde{\vlam} \to \vlambda^*$. However, given nonzero $\sigma^2$, $\mK$ is bounded away from $\mI$. Furthermore, since it is impossible to find a set of more than $M+1$ equidistant points in $\mathbb{R}^M$, the off-diagonal values of $\mK$ will differ when $M < N$, introducing distortion in the transformation.

However, even if $M < N$, if $\vy$ is a measurement from a \emph{random sensor group}, then the effect of this distortion can be mitigated such that $\widetilde{\vlam}$ is a reliable estimator of $\vlambda^*$ from a support recovery perspective:

\begin{theorem}
    \label{thm:small-groups}
    If $\vy \sim \mathcal{N}(\vphi_{n^*}^{(g)}, \sigma^2 \mI)$, $g$ is distributed uniformly, and $\vphi_{n}^{(g)} \overset{\mathrm{i.i.d.}}{\sim} \mathcal{N}(0, \mI)$, then if $G \to \infty$ and all elements of the maximizer $\widetilde{\vlam}$ of the Jensen bound are strictly positive, there exist $c_1 \geq 0, c_2 \in \mathbb{R}$ such that $\widetilde{\vlam}_n = c_1 \lambda_n^* + c_2$ for $1 \leq n \leq N$.
\end{theorem}

If $\widehat \vlambda$ has the same rank ordering as $\vlambda^*$, the exact support can be recovered. Therefore, we expect an increase in $G$ to improve performance in tasks such as support recovery.
From this result, however, we expect gains due to increasing $G$ to be less immediate than those due to increasing $M$ (and indeed, we see this in our simulations in Section~\ref{subsec:sumlambda_and_k}). To see this, contrast the asymptotic nature of Theorem~\ref{thm:small-groups} in $G$ with the fact that for a finite choice of $M$ (specifically $M=N$) we can select all $\vphi_n$ equidistant (or that for $M$ even smaller we can select $\PhiB$ satisfying a RIP with some acceptable distortion) and obtain the same reliability result.

\section{Simulations} \label{sec:sims}

\begin{figure*} 
  \centering
  \begin{tabular}{@{}l@{\hskip 0in}c@{}}
  \captionsetup[subfigure]{oneside,margin={0.7cm,0cm}}
  \subfloat[Measurements Required]{\label{fig:cs_M} \includegraphics[trim=0 35 13 50, clip, width=0.302 \textwidth, valign=t]{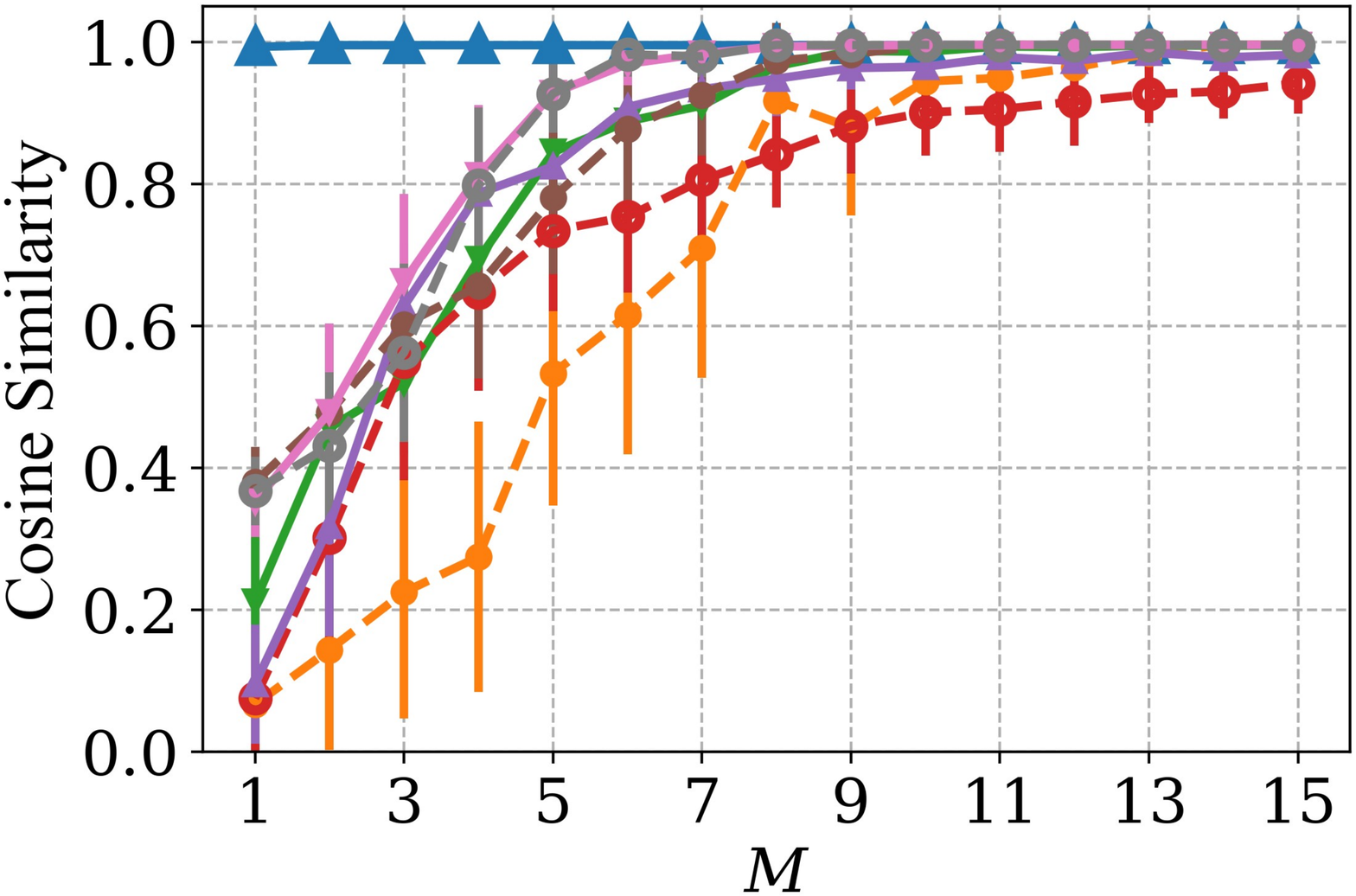}}
  \captionsetup[subfigure]{oneside,margin={0cm,0cm}}
  \subfloat[Tolerance to AWGN]{\label{fig:cs_sigma} \includegraphics[trim=115 35 -20 50, clip, width=0.27 \textwidth, valign=t]{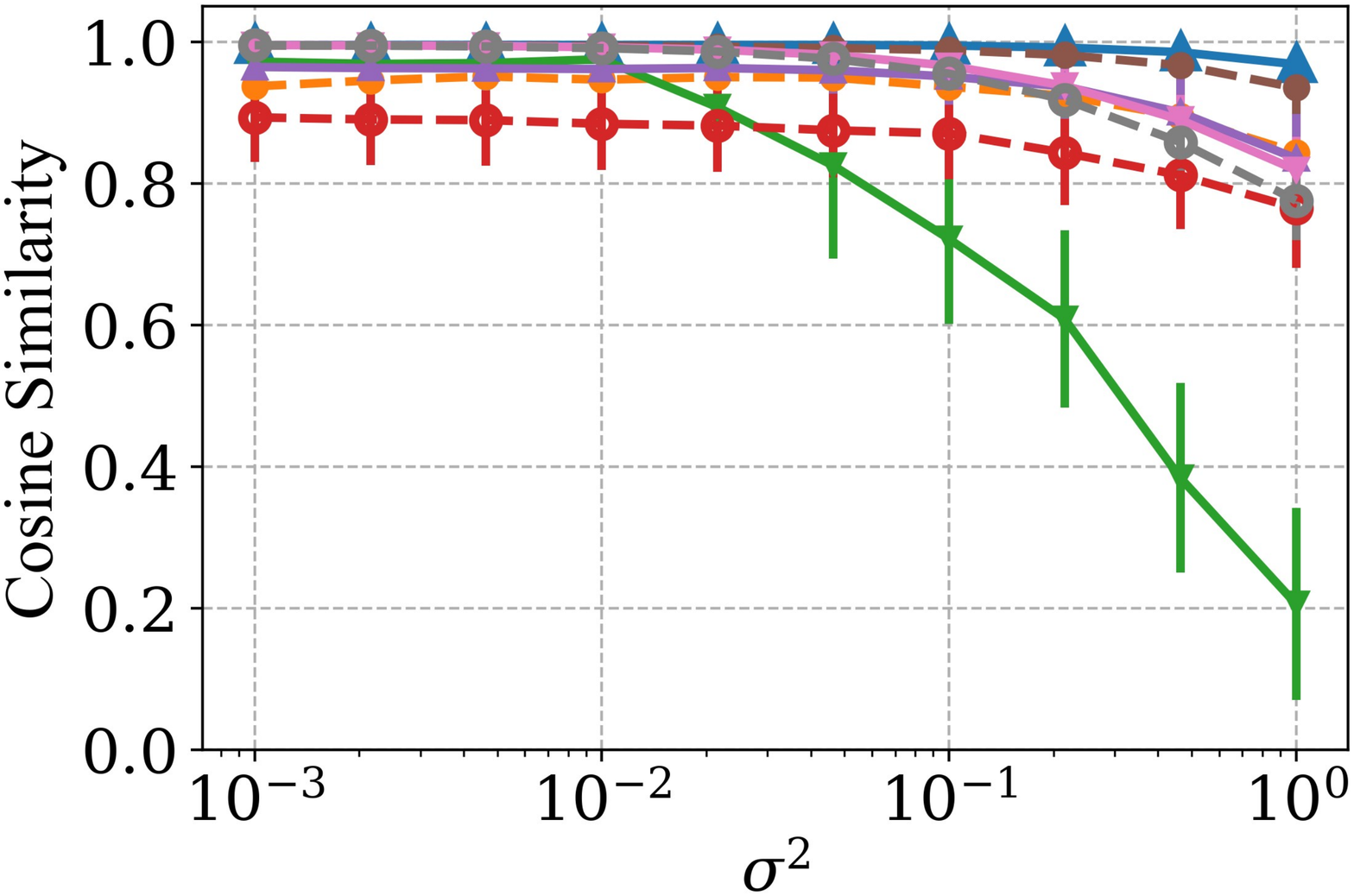}}
  \subfloat[Dynamic Range]{\label{fig:cs_lambda}\includegraphics[trim=115 30 0 45, clip, width=0.262 \textwidth, valign=t]{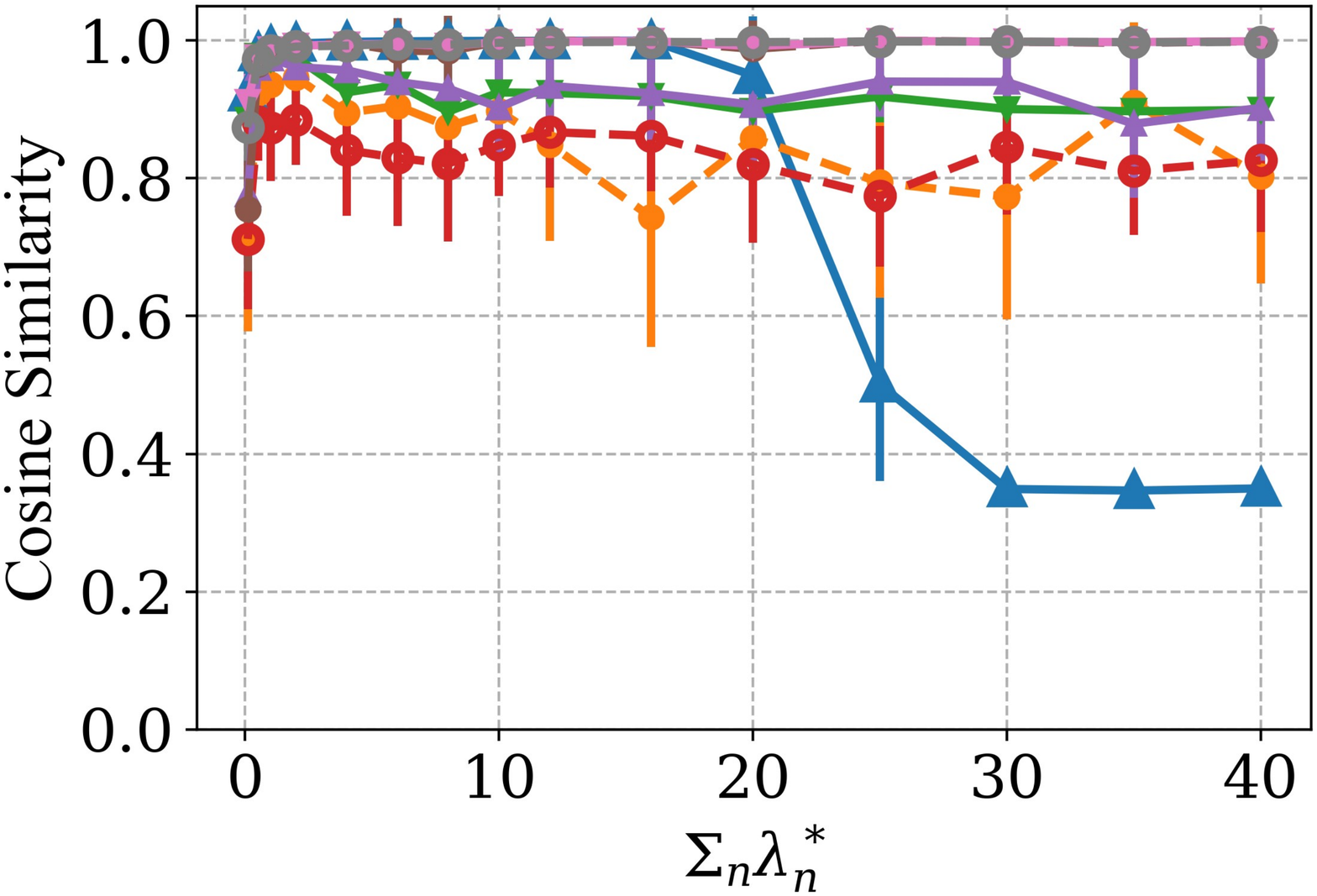}}
  \subfloat{\includegraphics[trim=0 0 0 22, clip, width=0.136\textwidth, valign=t]{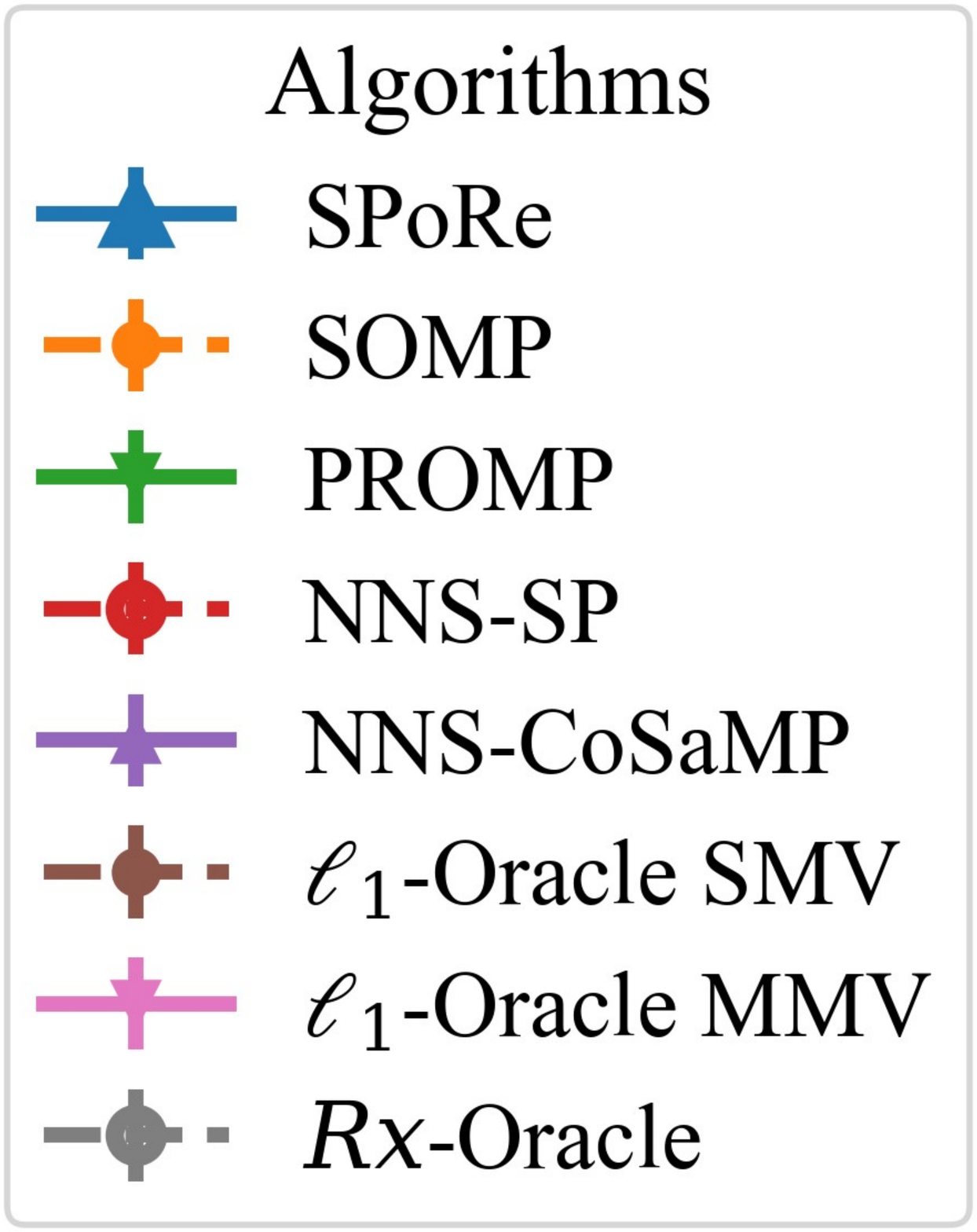}}
  \end{tabular}
\caption{Performance of SPoRe vs.\ compressed sensing baseline algorithms over 50 trials. Common settings unless otherwise specified are $M=10$, $k=3$, $N=20$, $D=100$, $G=1$, $\sum_n \lambda_n^*=2$. \textbf{(a)} Performance as a function of $M$, with $\sigma^2 = 10^{-6}$ for comparison in an effectively noiseless setting. \textbf{(b)} Performance as a function of AWGN variance, with $M=10$. \textbf{(c)} Performance as a function of $\sum_n \lambda_n^*$, with $\sigma^2 = 10^{-2}$ and $M=10$.}
\label{fig:cs_baselines}
\end{figure*}

In this section, we present comparisons of SPoRe against existing and custom baseline algorithms and follow with focused experimentation on SPoRe's performance and limitations. SPoRe and our custom alternating baseline (Section~\ref{subsec_sims_int}) are the only algorithms designed to output an estimate $\widehat{\vlambda}$ directly. For the algorithms that find an estimate $\widehat{\XB}$, we set their estimates $\widehat{\vlambda} = \frac{1}{D} \sum_{d=1}^D \widehat{\vx}_d$, i.e., the canonical Poisson MLE if $\XB^*$ were observed directly. For a performance metric, we chose cosine similarity between $\widehat{\vlambda}$ and $\vlambda^*$ as it captures the relative distribution of elements of the solution which we believe is of most utility to a user. Although comparisons of cosine similarity mask differences in magnitude, estimates with high cosine similarity also exhibited low mean-squared error in our experience (results not shown). We plot cosine similarity alone for brevity.
In all simulations, we use AWGN and set $\phi^{(g)}_{m,n} 
\overset{\mathrm{i.i.d.}}{\sim}
\mathrm{Uniform}(0,1)$ since many sensors are restricted to nonnegative measurements. For each parameter combination, we evaluate over 50 trials in which we draw new $\PhiB^{(g)}$ and $\vlambda^*$ for each trial. 
Due to high performance variability for some baseline algorithms, all error bars are scaled to $\pm \frac{1}{2}$ standard deviation for consistency and readability. 

\subsection{Comparison against existing baselines}
With no existing algorithm designed for Poisson signals, we compare against a number of algorithms with various relevant structural assumptions. We compare against both greedy and convex optimization approaches. First, we use DCS-SOMP~\cite{Dror09-DCS}, a generalization of the common baseline Simultaneous Orthogonal Matching Pursuit (S-OMP)~\cite{Tropp06-MMVpt1somp} that assumes no structure and greedily solves MMV problems for any value of $G$. Next, we use NNS-SP and NNS-CoSaMP~\cite{Kim16-MMVgreedyNN}, two greedy algorithms for nonnegative MMV CS motivated by subspace pursuit (SP)~\cite{Dai09_SP} and compressive sampling matching pursuit (CoSaMP)~\cite{Needell09_cosamp} which exhibited the best empirical performance in~\cite{Kim16-MMVgreedyNN}. For integer-based recovery, we use PROMP~\cite{Flinth17-PROMP}, an SMV algorithm for unbounded integer sparse recovery, to recover an estimate for each signal $\widehat{\vx}_d$. 

For comparison against best-possible performance of the baselines and to avoid hyperparameter search (for regularization weights, stopping criteria, etc.), we arm the baselines with relevant oracle knowledge of $\vlambda^*$ or $\XB^*$. While NNS-SP and NNS-CoSaMP require $k$ as an input, we also give DCS-SOMP and PROMP, algorithms that iteratively and irreversibly select support elements, knowledge of $k$ and have them stop after $k$ elements have been chosen. Additionally, we created two oracle-enabled convex algorithms. The $\ell_1$ norm is commonly used as a penalty for convex solvers to encourage sparsity in sparse recovery. Our $\ell_1$-Oracles include SMV and MMV versions, where in the MMV case, $\YB$ is collapsed to a single vector by summing $\sum_d \vy_d$, and a vector $\widehat{\vx}$ is recovered from which $\widehat{\vlambda} = \widehat{\vx}/D$. In~\cite{Tropp06-MMVpt2RX}, $\| \XB \|_{Rx}$ is suggested as a better alternative for MMV. Our $\ell_1$-Oracles and $Rx$-Oracle use 
$\sum_{n,d} x_{n,d} = \sum_{n,d} x_{n,d}^*$ and $\| \XB \|_{Rx} \leq \| \XB^* \|_{Rx}$ as convex constraints while minimizing $\sum_{g=1}^G \| \YB - \PhiB^{(g)} \XB^{(g)} \|_F$. We also set the affine constraint $\XB \geq 0$ for all three algorithms. We use the convex optimization package CVX in Matlab for these algorithms~\cite{cvx, gb08}.

From Fig.~\ref{fig:cs_M}, we see the crucial result that the $M < k$ regime is only feasible with SPoRe, while conventional CS algorithms, both SMV and MMV, fail. Such a result is expected; generally speaking, CS algorithms seek to minimize measurement error ($\|\YB-\PhiB \XB\|_F$) while constraining the sparsity of the recovered solution. CS theory focuses on $M > k$ since if $M < k$, $M \times k$ submatrices of $\PhiB$ yield underdetermined systems in general. In other words, there simply cannot be unique $k$-sparse minimizers of measurement error alone with $M < k$, so the conventional CS problem is not well-posed, unlike in the MMVP problem. Next, in Fig.~\ref{fig:cs_sigma}, we set $M=10$, a regime where most baselines performed nearly perfectly according to (Fig.~\ref{fig:cs_M}), and we increased the AWGN variance. We see that even in the conventional regime of $N > M > k$, SPoRe exhibits the highest noise tolerance which reflects the fact that its leverage of the Poisson assumption minimizes its dependence on accurate measurements.
Lastly, however, in Fig.~\ref{fig:cs_lambda}, SPoRe has the unique disadvantage of struggling to recover cases with high $\sum_n \lambda_n^*$. We observed that as $\sum_n \lambda_n^*$ increases, SPoRe's finite sampling results in few to no gradient steps taken as ``good'' samples with nonzero (numerically) $p(\vy|\vx_s)$ were drawn increasingly rarely, and SPoRe mistakenly terminates. Under AWGN, larger $\lambda^*_n$ raises the signal-to-noise ratio but can paradoxically compromise SPoRe's performance. If $M\gg k$ is a practical design choice, practitioners should consider existing MMV approaches if $\sum_n \lambda_n^*$ may be highly variable.

\subsection{Comparison against custom baselines: $M<k$} \label{subsec_sims_int}
In the $M < k$ regime, since with high probability we can bound the elements of $\XB^*$, we might expect the discrete nature of the problem to admit at the least a brute-force solution for obtaining $\widehat{\XB}$ that we can use to obtain $\widehat{\vlambda}$. Indeed, if measurement noise is low, then the integer signal that minimizes measurement error for $\vy_d$ is likely to be $\vx_d^*$. But a finite search space alone has not enabled integer-constrained CS research to achieve $M<k$ in general. 

One may wonder whether SPoRe is simply taking advantage of this practically finite search space and, by virtue of MC sampling over thousands of iterations, is effectively finding the right solution by brute force.
To address this possibility, we compare against an $\ell_0$-Oracle that is given $k$ and the maximum value in $\XB^*$ in order to test all $N \choose k$ combinations of $\XB$'s support. For each combination, it enumerates the $(\max(\XB^*)+1)^k$ possibilities for each $\vx_d$ and selects $\widehat{\vx}_d = \arg \min_{\vx} \| \vy_d - \PhiB \vx\|_2$. Finally, it returns the $k$-sparse solution with the lowest minimized measurement error. This algorithm is the only Poisson-free approach in this section. 

Comparing SPoRe and other Poisson-enabled baselines against the $\ell_0$-Oracle characterizes the effect of incorporating the Poisson assumption on recovery performance. An early solution of ours for tackling MMVP, which we now use as a baseline, was an alternating optimization framework to update estimates of $\widehat{\XB} = \argmax_\XB p(\XB|\YB, \widehat{\vlambda})$ and $\widehat{\vlambda} = \argmax_{\vlambda} p(\vlambda|\YB, \widehat{\XB})$. Noting that $p(\XB|\YB, \widehat{\vlambda}) \propto p(\YB|\XB)p(\XB|\widehat{\vlambda})$, this MAP framework for solving for $\XB$ under AWGN with variance $\sigma^2$ is
\begin{align}
  \widehat{\XB} & = \argmax_\XB \frac{1}{D} \sum_{d=1}^D \log P(\vy_d|\vx_d) + \log p(\vx_d|\widehat{\vlambda}) \label{eq:obj_alt_basic}\\
  &= \argmax_\XB \frac{1}{D} \sum_{d=1}^D \bigg[ -\frac{1}{2\sigma^2} \|\vy_d - \PhiB \vx_d\|^2_2 \nonumber \\ 
  &\quad +\sum_{n=1}^N \bigg(x_{n,d}\log \hat{\lambda}_n - \log \Gamma(x_{n,d}+1)\bigg)\bigg] \label{eq:obj_alt},
\end{align}
where the Gamma function $\Gamma(\cdot)$ is the continuous extension of the factorial ($\Gamma(x_{n,d}+1) = x_{n,d}!$) and is log-concave in the space of positive reals $\mathbb{R}_{++}^N$. We implemented the classic branch-and-bound (BB) algorithm~\cite{Land60-BB} to find the optimal integer-valued solution $\widehat{\XB}$ of the concave objective.
Once an estimate $\widehat{\XB}$ is available, the update to $\widehat{\vlambda}$ is also concave with the closed form solution $\widehat{\vlambda} = \frac{1}{D}\sum_d{\widehat{\vx}_d}$. 
The biconcavity of this objective function in $\XB$ and $\vlambda$ makes this approach attractive, but it is unclear how to best initialize $\widehat{\vlambda}$. 
We refer to this alternating baseline algorithm with the prefix~``Alt'' followed by the method of initialization. For example, for Alt-Random, we use random initialization with a small offset ($\hat{\lambda}_n \sim \mathrm{Uniform}(0,1)+0.1$) to avoid making any particular $\lambda_n$ irretrievable from the start.

We also explore a few ``guided'' initialization processes. The quantity $\sum_n \lambda_n^*$ can hypothetically be estimated from data if $P(\vx_d=\vzero|\vlambda)$ is significant and easily estimated from $\YB$.
In fact, quantification in microfluidics often relies on a clear identification of empty sample partitions (that is, where $\vx_d = \vzero$)~\cite{Basu17-digitalrev}. This motivates a strategy of relaxing the problem by optimizing $\XB$ with a Poisson assumption on the sum of each column $\sum_{n} x_{n,d}$ rather than each element of $\XB$ individually.
The $\sum_n \lambda_n^*$-Oracle is given $\sum_n \lambda_n^*$ and optimizes for $P(\sum_{n} x_{n,d} | \sum_{n} \lambda_n^*)$ in place of $P(\vx_d|\widehat{\vlambda})$ in~\eqref{eq:obj_alt_basic}. 
It is straightforward to show that this objective is also concave. Each estimate $\widehat{\vx}_d$ is solved via BB from which $\widehat{\vlambda} = \frac{1}{D} \sum_d \widehat{\vx}_d$. 
We use the $\sum_n \lambda_n^*$-Oracle as its own baseline and as an initialization to our alternating framework (Alt-$\sum_n \lambda_n^*$).
For the next guided initialization, we again use $\sum_n \lambda_n^*$ for an unbiased initialization where the first estimate of $\hat{\lambda}_n = (\sum_{n'} \lambda_{n'}^*)/N $ for all $n$ (Alt-Unbiased). Finally, we used the output of SPoRe as an initial value for $\widehat{\vlambda}$ (Alt-SPoRe). Alt-SPoRe can be understood as a way to use SPoRe to estimate $\widehat{\XB}$ if needed. 

\begin{figure}
  \centering
   \includegraphics[trim=0 122 0 115, clip, width=1 \linewidth]{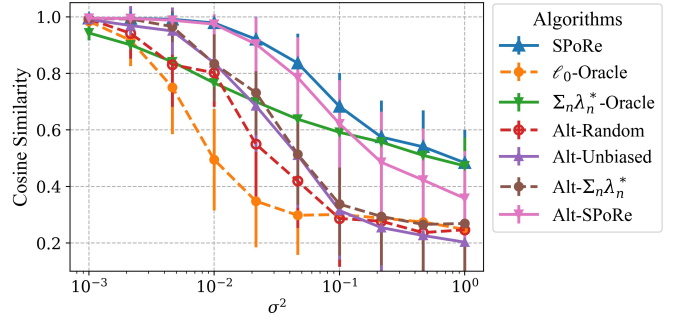} \caption{AWGN tolerance of integer-restricted algorithms over 50 trials with $M=2$, $k=3$, $N=10$, $D=100$, $G=1$, $\sum_n \lambda_n^*=2$.}
\label{fig:intbaselines}
\end{figure}

Fig.~\ref{fig:intbaselines} illustrates that SPoRe has the greatest tolerance to measurement noise whereas the $\ell_0$-Oracle has the least. This comparison illustrates the value of incorporating the Poisson assumption in recovery; specifically, the integer and sparsity structures (perfectly captured by the $\ell_0$-Oracle) are not sufficient for recovery under measurement noise. The alternating optimization algorithm's behavior was unexpected; initialization (other than with Alt-SPoRe) does not appear to have a major influence on performance. Surprisingly, comparing Alt-$\sum_n \lambda_n^*$ versus $\sum_n \lambda_n^*$-Oracle and Alt-SPoRe versus SPoRe, alternating seems to worsen the performance under high noise. Our interpretation is that in high noise settings, the ability of SPoRe to not ``overcommit" to a particular solution $\widehat{\vx}_d$ may be especially effective when $\vlambda^*$ is the signal of interest rather than $\XB^*$. Any given measurement $\vy_d$ may make the specific estimate $\widehat{\vx}_d$ arbitrarily unreliable. In our alternating framework with $\widehat{\vx}_d$ recovered separately for each $d$, errors on individual estimates accumulate. SPoRe instead makes gradient steps based on batches of observations, helping it maintain awareness of the full distribution of measurements.

\subsection{Sparsity and $\sum_n \lambda_n^*$} \label{subsec:sumlambda_and_k}

\begin{figure*} 
  \centering
  \captionsetup[subfigure]{oneside,margin={0.7cm,0cm}}

  \subfloat[Initial $\widehat{\vlambda} = \mathbf{0.1}$]{\label{fig:klam_reg_init} \includegraphics[trim=15 40 13 45, clip, width=0.29\textwidth, valign=t]{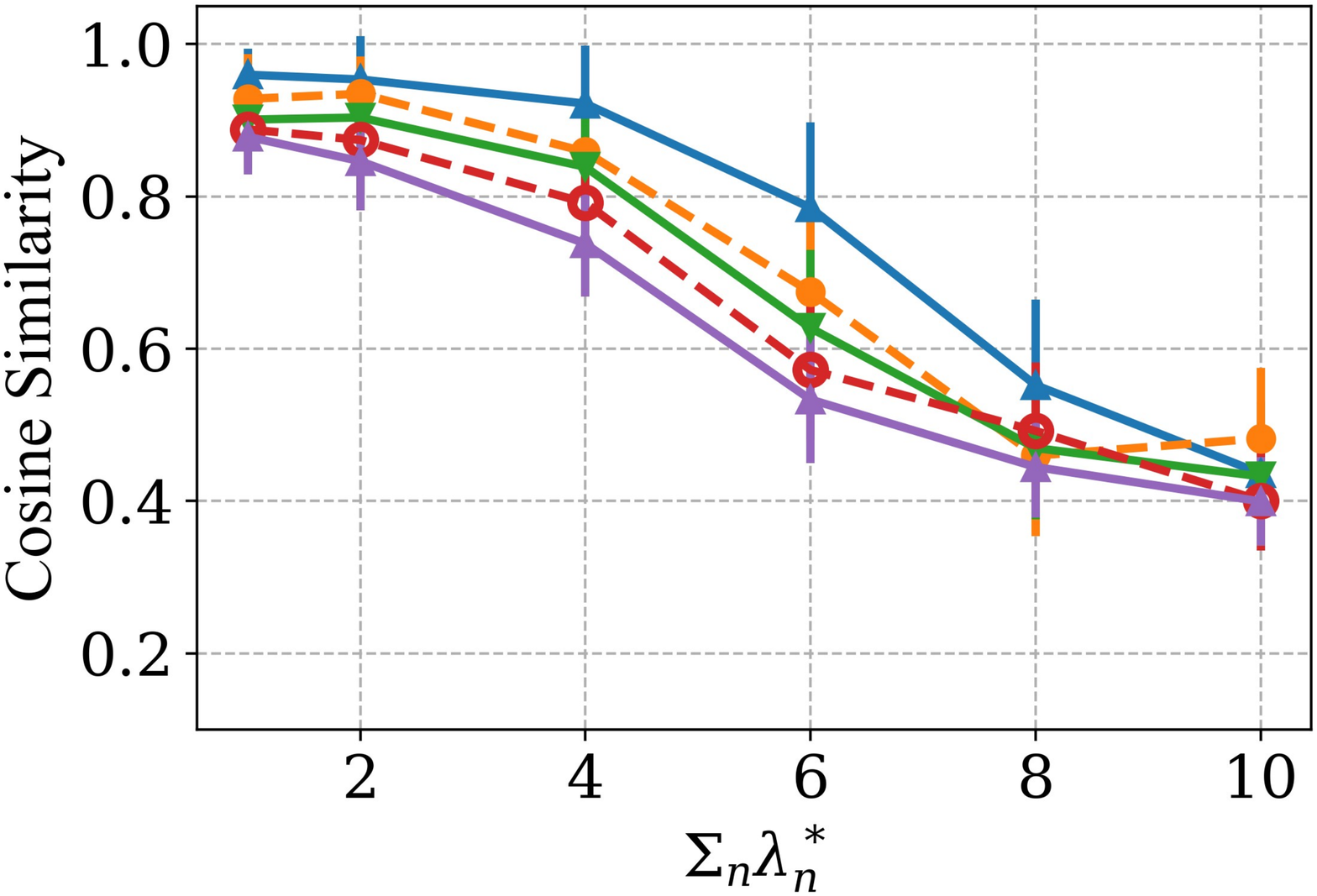}}
  \hfill
  \subfloat[Initial $\widehat{\vlambda} \approx \vlambda^*$]{\label{fig:klam_lamStar_init} \includegraphics[trim=15 40 13 45, clip, width=0.29\textwidth, valign=t]{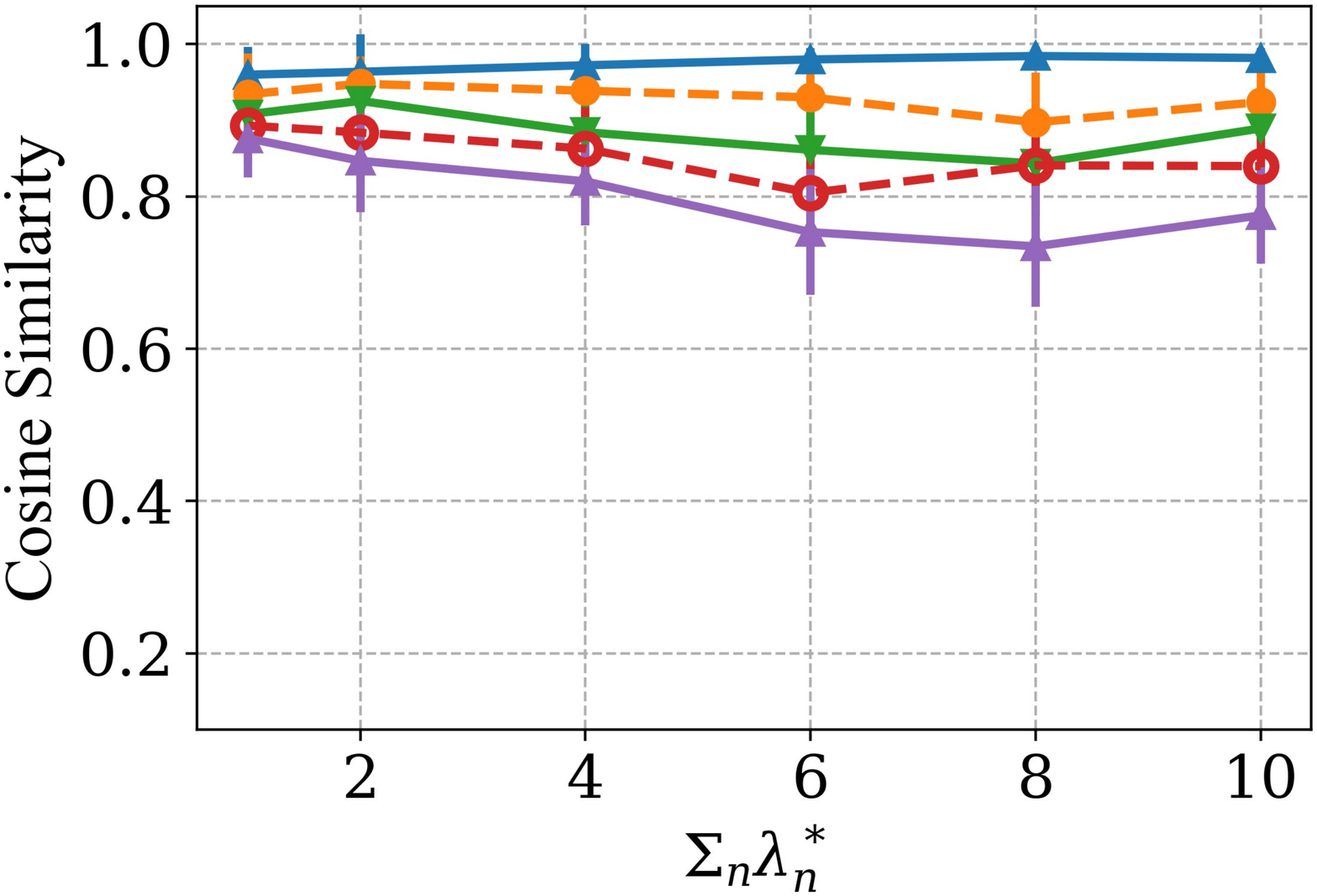}}
  \hfill
  \subfloat[Variance of MC Gradients] {\label{fig:klam_partials} \includegraphics[trim=15 40 13 50, clip, width=0.29\textwidth, valign=t]{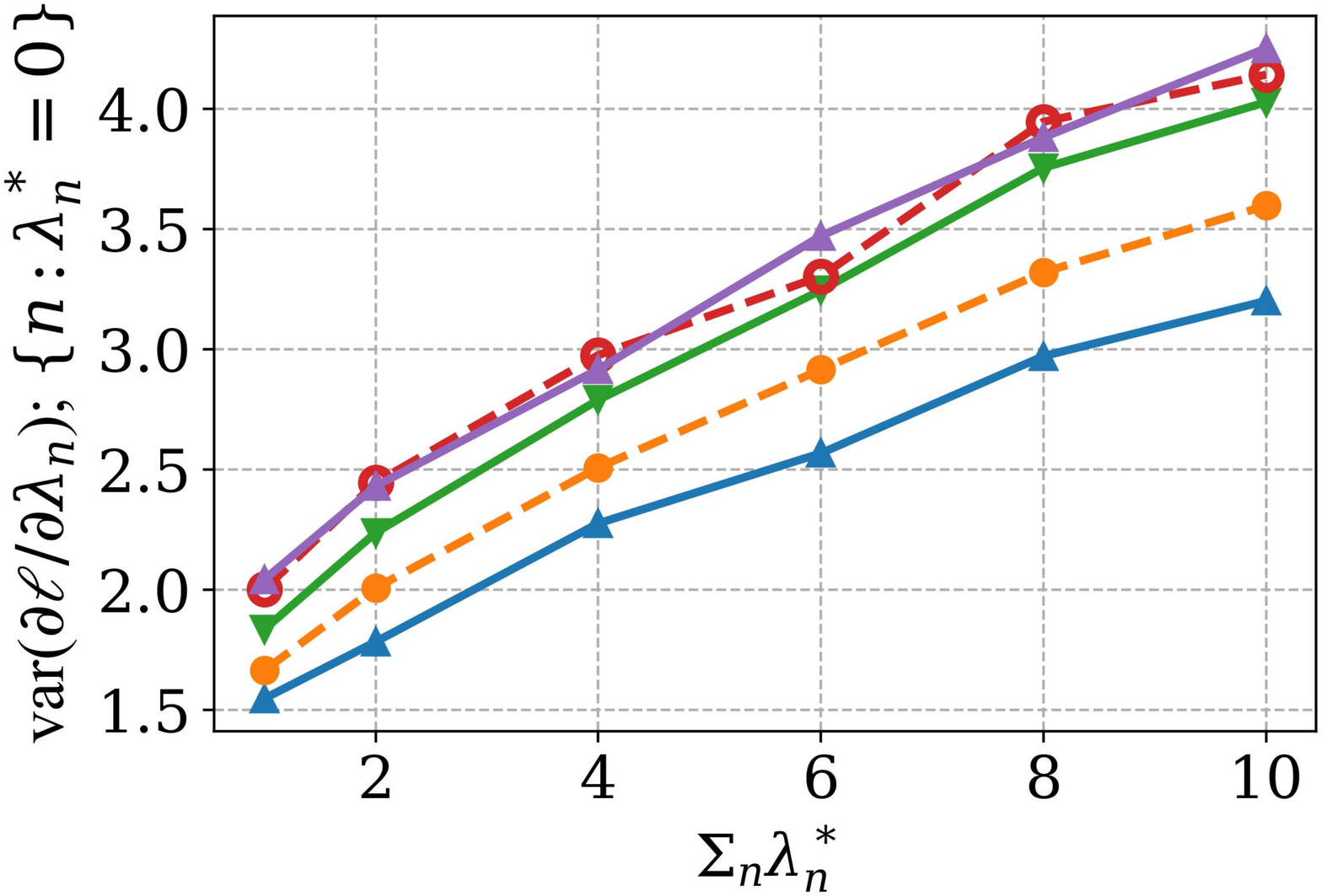}}
  \hfill
  \subfloat{\includegraphics[trim=115 0 110 140, clip, width=0.08\textwidth, valign=t]{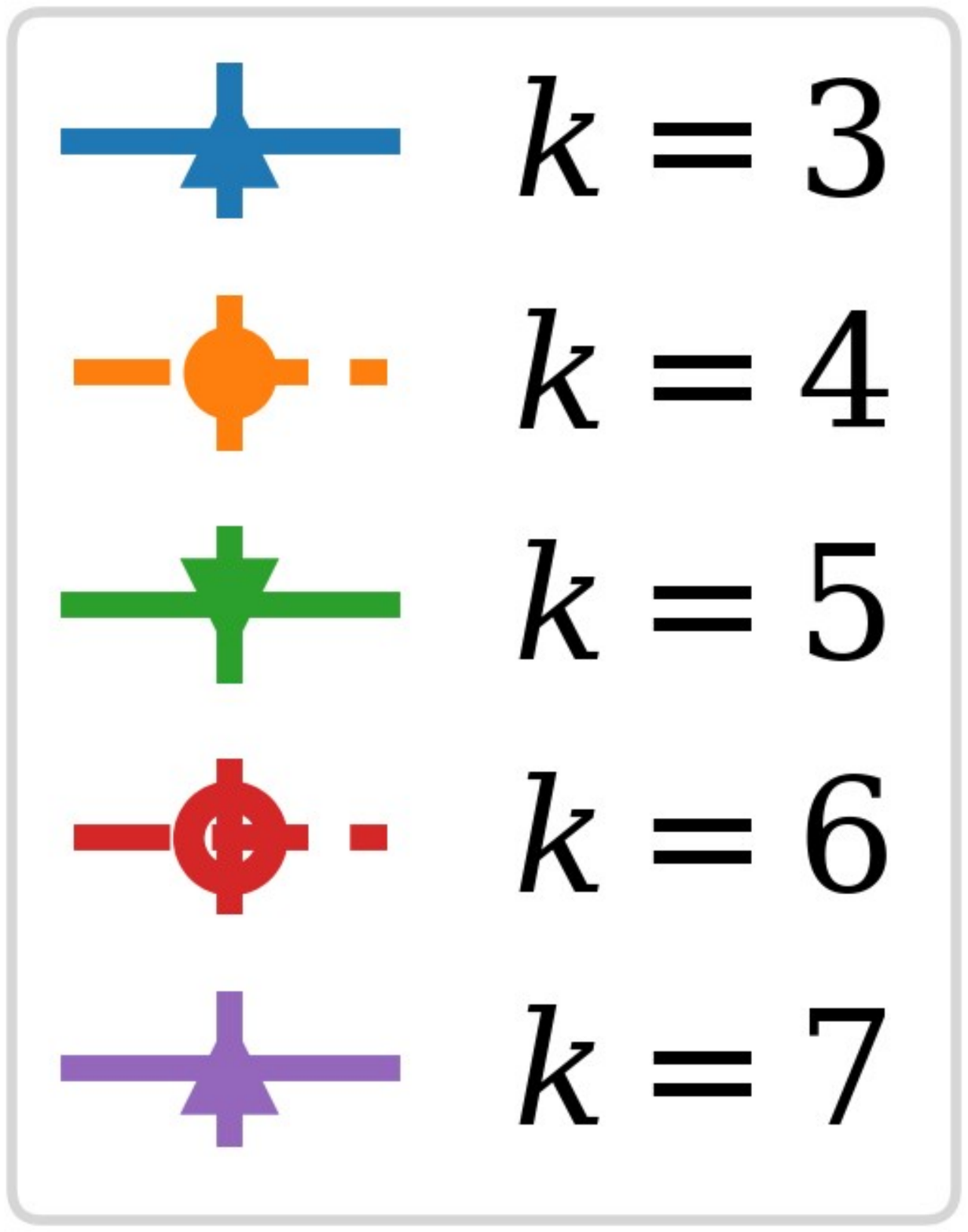}}\\
\caption{SPoRe's performance and behavior as a function of $k$ and $\sum_n \lambda^*_n$ over 50 trials. Common settings unless otherwise specified are $M=2$, $N=50$, $D=1000$, $G=1$, $\sigma^2 = 10^{-2}$. \textbf{(a)} Performance when initialized with standard $\widehat{\vlambda}=\mathbf{0.1}$. \textbf{(b)} Performance when initialized with $\widehat{\vlambda} \approx \vlambda^*$, specifically $\hat\lambda_n = \max \{\epsilon, \lambda_n^*\}$. \textbf{(c)} Average variance of partial derivatives for indices $n \notin \mathrm{supp}(\vlambda^*)$ evaluated at $\widehat{\vlambda} \approx \vlambda^*$.} \label{fig:klam}
\end{figure*}

We empirically tested the limitations of SPoRe's recovery performance under very challenging conditions of $M=2$, ${3 \leq k \leq 7}$, $N=50$, $\sigma^2 = 10^{-2}$.
Here we set $D=1000$ to better reflect the typical capabilities of biomedical systems, whereas $D=100$ in our baseline comparisons was due to our budget on computational time strained by solving BB for $\widehat{\vx}_d$. From our analysis and previous simulations, we expect that both $k$ and the magnitudes of $\lambda^*_n$ will influence recovery. Fig.~\ref{fig:klam} probes when and why SPoRe fails. Fig.~\ref{fig:klam_reg_init} illustrates SPoRe's performance decreases with increasing $k$ and $\sum_n \lambda^*_n$. To elucidate the cause of poor performance, Fig.~\ref{fig:klam_lamStar_init} shows SPoRe's performance under the same conditions when initialized at the optimum. SPoRe's maintenance of high cosine similarity in this case means that in Fig.~\ref{fig:klam_reg_init}, SPoRe is converging to incorrect optima (or terminating before convergence). These two figures depict fundamental limitations of stochastic optimization in a challenging landscape.

Moreover, Fig.~\ref{fig:klam_partials} illustrates that MC gradients decrease in quality with high $\sum_n \lambda^*_n$ and $k$. In SPoRe, recall that we set a minimum $\hat{\lambda}_n \geq \epsilon = 10^{-3}$ so that nonzero $x_n^*$ have a chance of being sampled for all $n$. We keep $S$ fixed as we increase $k$ and $\sum_n \lambda^*_n$, and we see that the variance of the gradient increases at coordinates where $\hat\lambda_n = \epsilon$ and $\lambda_n^* = 0$. 
Such an effect accounts for some drift from the optimum observed in Fig.~\ref{fig:klam_lamStar_init} that increases with $k$, and we believe that it helps to explain the inability to converge to the optimum in Fig.~\ref{fig:klam_reg_init}. Future work can explore alternative techniques for stochastic optimization and sampling. Practitioners may benefit significantly from reducing $\sum_n \lambda^*_n$ if faced with limitations in $M$.

However, note in~\eqref{eq:grad_final} that SPoRe's gradients are defined by an average of $\vx_s$ weighted by $p(\vy_d|\vx_s)$. The previous result from Fig.~\ref{fig:cs_lambda} in which SPoRe performed well with $\sum_n \lambda_n^* \leq 20$ when $M=10$ illustrates that limitations of MC sampling may be offset by improving the ability of $p(\vy|\vx_s)$ to guide gradients. In Fig.~\ref{fig:Gs}, we explore this notion further for $M$-constrained systems by increasing $G$. One may wonder how increasing $G$ compares to a CS problem with $GM$ measurements (i.e., $\Bar{\PhiB} \in \mathbb{R}^{GM \times N}$). Although $\vlambda^*$ is fixed across groups, the $\vx_d$ are random such that there is no reasonable method to directly stack individual measurements $\vy_d$ from multiple groups. Instead, we created a new baseline $\ell_1$-Oracle $GM$ SMV. Denote the average of measurements and signals in each group $g$ as $\Bar{\vy}^{(g)}$ and $\Bar{\vx}^{(g)}$, respectively. Our new baseline stacks all $\Bar{\vy}^{(g)}$ into $\Bar{\vy} \in \mathbb{R}^{GM \times 1}$ measurement vector. Given $\sum_{n,d} x_{n,d}^*$, the algorithm then directly recovers $\widehat{\vlambda}$ with a convex process similar to that of the previously described $\ell_1$-Oracle SMV. Stacking measurements implicitly assumes that for each group, $\Bar{\vy}^{(g)} \approx \PhiB^{(g)}\vlambda^*$, or that $\Bar{\vx}^{(g)} \approx \vlambda^*\; \forall g$. It can be easily shown that the relative errors in these approximations reduce with increasing $D$ or $\vlambda^*$.

Although increasing $D$ is feasible in microfluidics, it generally corresponds with a reduction in $\vlambda^*$ since a sample's total analyte content is fixed. Therefore, in Fig.~\ref{fig:Gs}, we focus on the influence of the magnitude of $\vlambda^*$. 
In Fig.~\ref{fig:G_lam10}, we used the most challenging settings from Fig.~\ref{fig:klam} with $k=7$ and $\sum_n{\lambda^*_n}=10$. As expected from our analysis in Section~\ref{subsec:small}, larger choices of $M$ make SPoRe much more effective per sensor group, but near perfect recovery is achievable even with $M=1$. However, note that the new oracle baseline performs almost identically to SPoRe, with SPoRe exhibiting a modest advantage only when $GM$ is comparable to or less than $k$. When we reduce $\sum_n \lambda_n^*$ to $1$ in Fig.~\ref{fig:G_lam1}, the assumption that  $\Bar{\vx}^{(g)}\approx \vlambda^*$ becomes far less valid. As a result, the performance improvement with SPoRe is dramatic. 
For instance, what SPoRe achieves with $M=1$ is only matched by the oracle baseline when $M=3$.
For applications in which $\Bar{\vx}^{(g)} \approx \vlambda^*$ and $GM > k$, practitioners could consider reformulating the recovery problem as a standard CS problem. However, SPoRe is uniquely suited for systems with $GM<k$ and is the best generalized approach for applications with lower $\vlambda^*$ or $D$.

\begin{figure}[]
  \centering
  \subfloat{\label{fig:G_leg}
      \includegraphics[trim=-100 250 0 250, clip, width=0.7 \linewidth, valign=t]{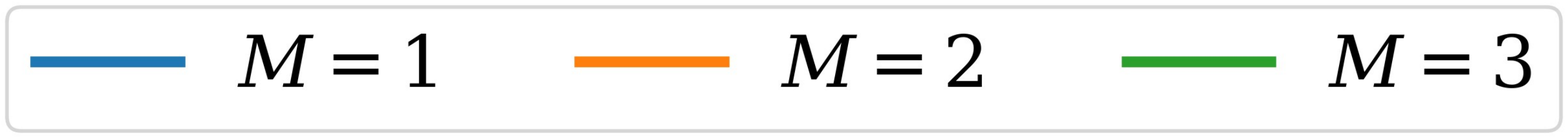}}
  
  \rule{0pt}{0.1ex}   
  
  \setcounter{subfigure}{0}
  \begin{tabular}{@{}l@{\hskip 0in}c@{}}
      \captionsetup[subfigure]{oneside,margin={0.8cm,0cm}}
      \subfloat[$\sum_n \lambda_n^* = 10$]{\label{fig:G_lam10}
        \includegraphics[trim=0 35 0 50, clip, width=0.523 \linewidth, valign=t]{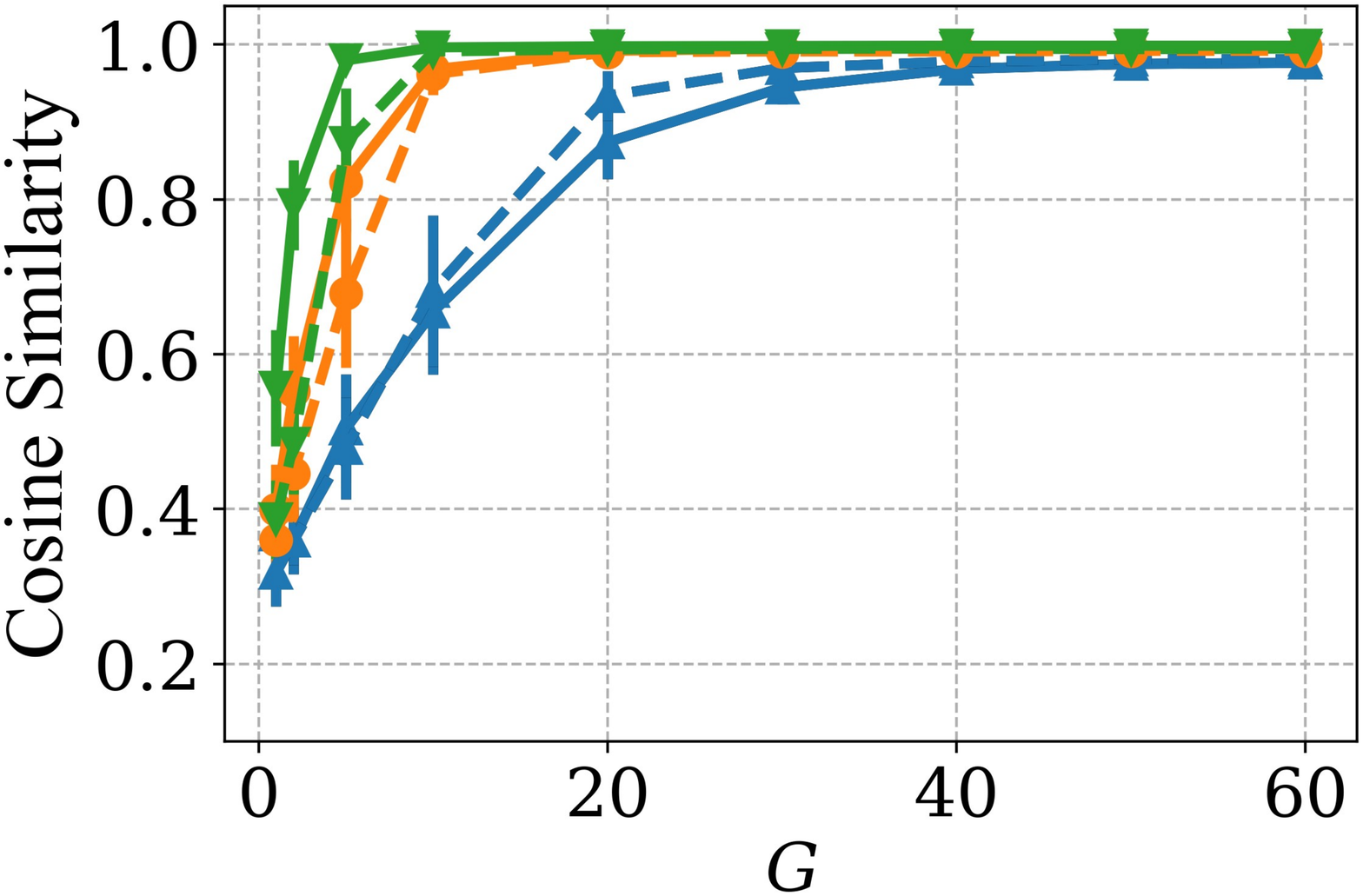}}
      &
      \captionsetup[subfigure]{oneside,margin={0cm,0cm}}
      \subfloat[$\sum_n \lambda_n^* = 1$]{\label{fig:G_lam1} \includegraphics[trim=125 35 0 50, clip, width=0.437 \linewidth, valign=t]{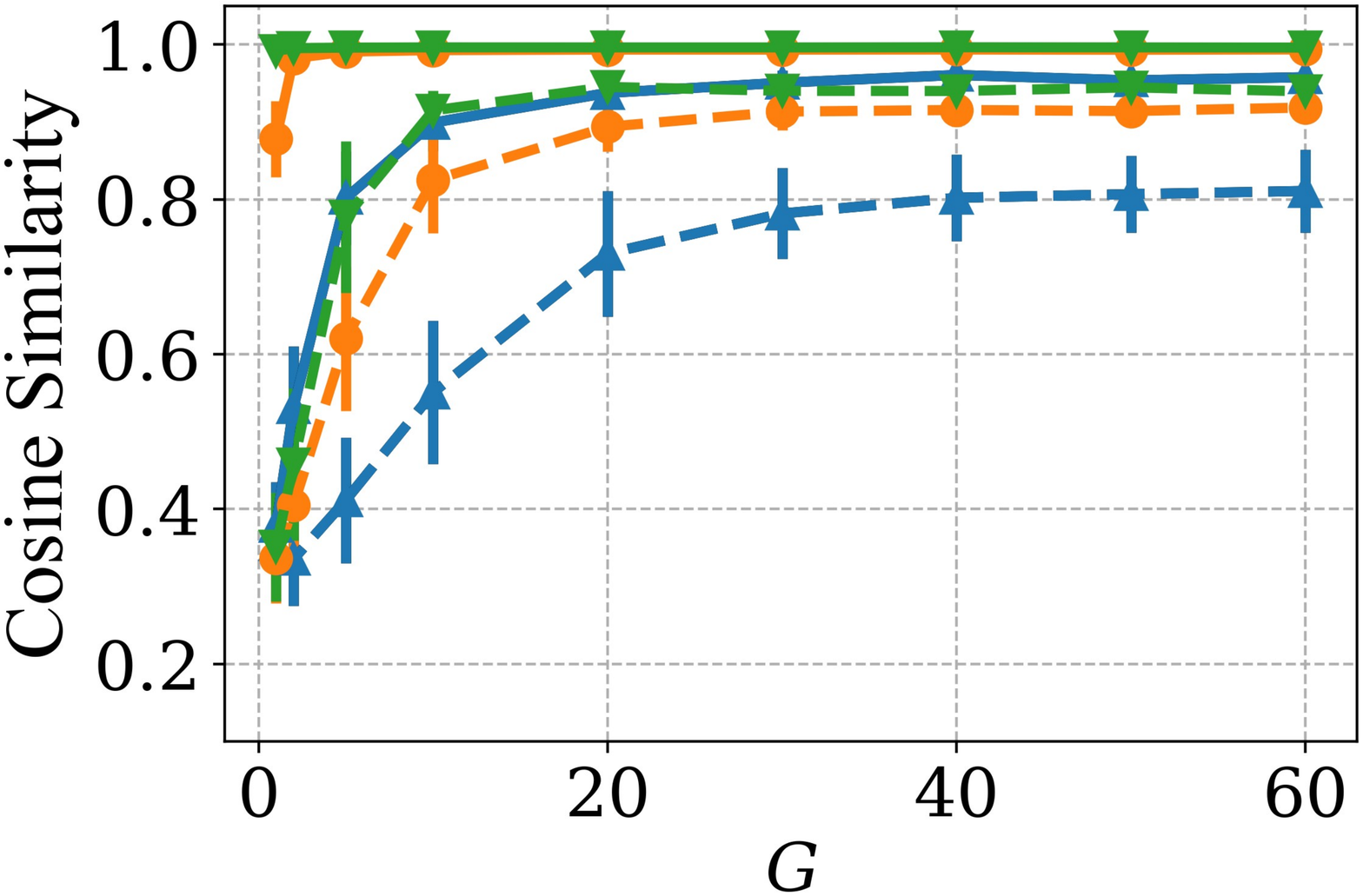}}
  \end{tabular}
  
\caption{Performance of SPoRe (solid) vs.\ $\ell_1$-Oracle $GM$ SMV (dashed) as a function of $G$. Common settings are $k=7$, $N=50$, $D=1000$, $\sigma^2 = 10^{-2}$. \textbf{(a)} Comparison with $\sum_n \lambda_n^* = 10$, motivated by Fig.~\ref{fig:klam_reg_init}. \textbf{(b)} Comparison with $\sum_n \lambda_n^* = 1$.}
\label{fig:Gs}
\end{figure}

\subsection{Efficiency} \label{subsec:efficiency}

For system design, it is helpful to know how many observations $D$ are necessary and sufficient for stable estimation of $\widehat{\vlambda}$. Such insight is often derived from the analysis of Fisher Information $\mathcal{I}$. Recall that for direct observations of Poisson signals $\vx^*_d$, the ideal case, the MLE solution $\hat{\lambda}_n = \sum_d x_{n,d}^* /D$ is an \emph{efficient estimator} and achieves the \textit{Cramér--Rao} bound such that $\mathrm{var}(\hat{\lambda}_n) = \lambda^*_n/D$.

In Section~\ref{subsec:fisher}, we derived the reduction in $\mathcal{I}_{n,n}$ from MMVP measurements and found reason to expect that the reduction increases with $k$. Here, we empirically characterize this effect (Fig.~\ref{fig:efficiency}). The matrix $\mathcal{I}$ is evaluated at $\vlambda^*$, so we only consider $n \in \mathrm{supp}(\vlambda^*)$. In MMVP, the variance of $\hat{\lambda}_n$ will depend on $\PhiB$ and $n$, but by redrawing random $\PhiB$ and $\vlambda^*$ over 50 trials, we hope to smooth out these dependencies and capture the broader effect of low dimensional projections. For a concise comparison considering $n \in \mathrm{supp}(\vlambda^*)$, we set all $\lambda^*_n = 1$, pool all of $\hat{\lambda}_n$ across all trials, and compute a single average variance for each $k$ and $D$. Because the Fisher Information describes the optimization landscape near the optimum, we chose parameter settings ($M=2$, $G=20$) based on our results in Fig.~\ref{fig:G_lam10} to be confident that SPoRe is arriving near the optimum and the estimation variance is not confounded by poor estimates. 

Fig.~\ref{fig:efficiency} shows a noticeable increase in the estimation variance and verifies that this deviation from the ideal bound is exponentially worsened in $k$. However, we argue that the variance quickly becomes negligible at reasonable $D$ for practical purposes. 
Practitioners could consider the necessary precision of estimation and the maximum expected $k$ for an application, increase $D$ as needed, and worry little about the influence of noisy measurements in low dimensions.

\begin{figure} 
  \centering
  \includegraphics[trim=0 65 0 45, clip, width=\fw \linewidth]{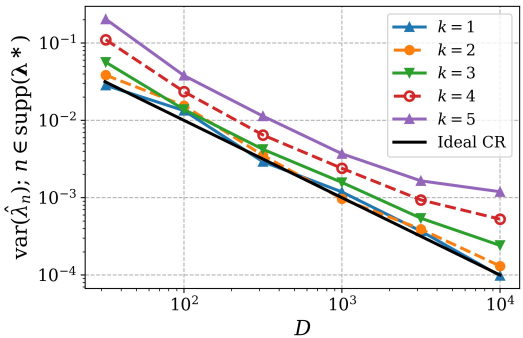}
  \caption{Comparison of average variance of $\hat{\lambda}_n$ from SPoRe versus the ideal Cramér--Rao (CR) bound as a function of $D$ over 50 trials with $n \in \mathrm{supp}(\vlambda^*)$, $\lambda^*_n=1$, $M=2$, $N=50$, $G=20$, $\sigma^2 = 10^{-2}$.} \label{fig:efficiency}
\end{figure}

\section{Discussion}

We have found that the structure in the MMVP problem can be easily exploited for substantial improvements in signal recovery. While compressed sensing of arbitrary integer signals has proven challenging in the past, Poisson constraints not only make the recovery problem tractable, but even significantly easier. Most inverse problems necessitate constraints that make the signal-to-measurement transformations nearly isometric: in compressed sensing, these manifest as restrictions on $\mathbf{\Phi}$, noise, and the relationship between $M$, $N$, and $k$. In MMVP, recovery of $\vlambda^*$ is theoretically possible under very lax conditions on $\mathbf{\Phi}$ (Theorem~\ref{th:id_weights}) and practically achievable as shown in our simulations. 

In practice, our new SPoRe algorithm exhibits high performance even under deliberately challenging circumstances of high noise and $M<k$. 
Because the log-likelihood is not concave, SPoRe's gradient ascent approach is not theoretically guaranteed to find a global optimum since local optima may exist. However, if they exist, we speculate that SPoRe is naturally poised to evade these traps due to stochasticity in its gradient steps from both batch draws and MC integrations. 

We noted a few scenarios in which SPoRe's MC sampling appears to cause issues with convergence or early termination that are generally associated with increases in $k$ and $\sum_n \lambda^*_n$. We anticipate that further increases in $N$ may also contribute to these effects. While $k$ and $N$ are entirely determined by the application, system designers can reduce $\sum_n \lambda^*_n$ by increasing the spatial or temporal sampling rate. In microfluidics, this adjustment translates to either generating more (smaller) partitions $D$ given a fixed sample volume or diluting a sample prior to  partitioning. 
Our initial implementation of SPoRe uses $S=1000$, can easily run on personal computers, and is sufficient for systems with $N<10^2$. This scale is appropriate for most applications in biosensing, and future work with parallelized or adaptive sampling strategies could improve the reliability of recovery for larger systems. Moreover, we found that increasing $M$ and $G$ appear to mitigate poor performance due to excessive sampling noise.

The ability to recover signals with $M<k$ in MMVP, even in the extreme case of $M=1$, is unprecedented in CS and offers a new paradigm for sensor-constrained applications. The state-of-the-art method for achieving this efficiency in microfluidics is to essentially guarantee single-analyte capture for classification by substantially increasing the sampling rate, whereas our MMVP framework is not reliant on such an intervention. Increasing $G$ can make SPoRe reliable under harsh conditions, is straightforward with microfluidics, and offers a potent alternative to adjusting the sampling rate. For example, diluting to a tolerable concentration is challenging with samples of unknown content such as in diagnostics applications.
Our group is continuing research in CS-based microbial diagnostics \cite{Aghazadeh16-UMD} by working towards an \textit{in vitro} demonstration of MMVP-based diagnostics with microfluidics. 

Our initial theoretical and empirical results show the promise of MMVP, but there are many directions for further research. For instance, theoretical results that precisely relate $M$, $N$, $D$, $k$, $\lambda^*$ and noise such as in a recovery guarantee would be highly valuable. Moreover, SPoRe can accept any signal-to-measurement model $p(\vy|\vx)$. While we have proven identifiability under linear mappings with common noise models, SPoRe can be easily applied with any application-specific model\footnote{Our full code base is available at \url{https://github.com/pavankkota/SPoRe} with instructions on how to implement SPoRe with custom models.} even if proving identifiability of the Poisson mixture is difficult. With growing interest in microfluidics, SPoRe's promising performance in the MMVP problem warrants further research to ensure that the statistical assumptions underlying these new technologies are leveraged~to~their~full~potential.

\appendix
\subsection{Proof of Theorem~\ref{th:id_weights}} \label{app:theo}

We use a direct proof, assuming $P(\mathcal{C}_\vu|\vlambda) = P(\mathcal{C}_\vu|\vlambda') \ \forall \vu$ and proving the resulting implication $\vlambda = \vlambda'$. Let $\vz^{(\vx)} \in \mathbb{Z}^N$ be such that $\vx + \vz^{(\vx)} \in \mathcal{C}_{\vx}$. By Definition~\ref{defn_coll}, $\vx + \vz^{(\vx)} \in \mathbb{Z}_+^N$ and $\vz^{(\vx)} \in \mathcal{N}(\PhiB)$. 

\begin{lemma} \label{lem:sum_lambda}
If $\mathcal{N}(\mathbf{\Phi}) \cap \mathbb{R}_+^N = \{\mathbf{0} \}$, and $P(\mathcal{C}_\mathbf{0}|\vlambda) = P(\mathcal{C}_\mathbf{0}|\vlambda')$, then $\sum_n \lambda_n = \sum_n \lambda_n'$.
\end{lemma}

\begin{proof} The null space condition on $\mathbf{\Phi}$ means that $\mathcal{C}_\mathbf{0} = \{ \mathbf{0} \}$; there is no vector $\vz^{(\mathbf{0})}$ that satisfies $\mathbf{0} + z^{(\mathbf{0})} \in \mathbb{Z}_+^N$ other than $\vz^{(\mathbf{0})} = \mathbf{0}$. Therefore, $P(\mathcal{C}_\mathbf{0}|\vlambda) = P(\mathcal{C}_\mathbf{0}|\vlambda') \Rightarrow e^{-\sum_n \lambda_n} = e^{-\sum_n \lambda_n'} \Rightarrow \sum_n \lambda_n = \sum_n \lambda_n'$. 
\end{proof}

We now turn our attention to the \textit{one-hot} collision sets. Let $\ve_j$ denote the $j$th standard basis vector.
By Definition~\ref{defn_coll}, 
$\mathcal{C}_{\ve_j} = \{ \vx: \PhiB \vx = \vphi_j, \vx \in \mathbb{Z}_+^N \}$. For $\mathcal{C}_{\ve_j}$
that contain \textit{only} $\ve_j$, we have the following result: 

\begin{lemma} \label{lem:onehot}
If $\mathcal{N}(\mathbf{\Phi}) \cap \mathbb{R}_+^N = \{\mathbf{0} \}$ and $\mathcal{C}_{\ve_j} = \{ \ve_j \}$, then ${\lambda}_j = \lambda'_j$.
\end{lemma}

\begin{proof} The restriction on $\mathcal{C}_{1j}$ means $P(\mathcal{C}_{1j}|\vlambda) = P(\mathcal{C}_{1j}|\vlambda') \Rightarrow \lambda_j e^{-\sum_n \lambda_n} = \lambda'_j e^{-\sum_n \lambda_n'}$. Applying Lemma~\ref{lem:sum_lambda} yields $\lambda_j = \lambda'_j$. 
\end{proof}

By similar arguments to Lemmas~\ref{lem:sum_lambda} and~\ref{lem:onehot}, we can prove Corollary~\ref{cor:id_continuous} under the assumption that there are no collisions instead of the null space condition.  When the elements of $\PhiB$ are independently drawn from continuous distributions, the collision of any particular $\vx$ and $\vx'$ occurs with probability zero. Since $\mathbb{Z}_+^N \times \mathbb{Z}_+^N$ is countably infinite, there are no collisions almost surely. As such, $\setC_{\ve_j} = \{\ve_j\} \forall j$, and therefore $\vlambda = \vlambda'$.

For the conditions in Theorem~\ref{th:id_weights}, the following Lemma states the existence of at least one $j$ satisfying
$\mathcal{C}_{\ve_j} = \{\ve_j\}$:

\begin{lemma} \label{lem:exist_onehot}
If $\mathcal{N}(\mathbf{\Phi}) \cap \mathbb{R}_+^N = \{\mathbf{0} \}$ and $\vphi_n \neq \vphi_{n'} \ \forall n, n' \in \{1, ... , N\}$ with $n \neq n'$, then $\exists \ j$ such that 
$\mathcal{C}_{\ve_j} = \{ \ve_j \}$. 
\end{lemma}

\begin{proof} If 
$\mathcal{C}_{\ve_j} = \{\ve_j\}$,
then $\nexists \ \vz^{(\ve_j)} \neq \vzero$. Define $P$ as the number of one-hot collision sets that contain more than just
$\ve_j$, and note that $P \leq N$.
Without loss of generality, let us say that $\setC_{\ve_j}$ for $j \in \{1, ..., P\}$ meet this condition. Lemma~\ref{lem:exist_onehot} effectively says that $P < N$, such that $N-P > 0$ one-hot collision sets contain only 
$\ve_j$. 
We proceed with a proof by contradiction by assuming $P = N$. 

By our null space condition, 
$\vz^{(\ve_j)}$ 
must contain both positive and negative integers. There are two additional conditions on nontrivial $\vz^{(\ve_j)}$. First, because $\ve_j + \vz^{(\ve_j)} \in \mathbb{Z}_+^N$, the only negative component of $\vz^{(\ve_j)}$ is $z^{(\ve_j)}_j=-1$. To see this, if $z^{(\ve_j)}_i$ for $i \neq j$ were negative, then $\ve_j + \vz^{(\ve_j)}$ would be negative at index $i$, and if $\vz^{(\ve_j)}_j$ were less than $-1$, then $\ve_j + \vz^{(\ve_j)}$ would be negative at index $j$. Second, $\vz^{(\ve_j)}$'s 
positive elements must total to at least 2. A single positive element of $z_i^{(\ve_j)} = 1$ would imply that $\vphi_i = \vphi_j$, violating a condition on $\PhiB$.

With $P=N$, let us concatenate the 
$\vz^{(\ve_j)}$ 
column vectors into a matrix for visualization. 
\begin{equation} \label{Zdiag} 
    \mathbf{Z} = 
    \begin{bmatrix}
    -1 & z^{(\ve_2)}_1 & \ldots & z^{(\ve_N)}_1\\ 
    z^{(\ve_1)}_2 & -1 & \ldots & z^{(\ve_N)}_2 \\
    \vdots & \vdots & \ddots & \vdots \\
    z^{(\ve_1)}_N & z^{(\ve_2)}_N & \ldots & -1 
    \end{bmatrix}.
\end{equation}
Note that each column 
$\vz^{(\ve_j)}$ 
in this matrix is symbolic for any vector that satisfies the conditions we described. All columns of $\mathbf{Z}$ are in $\mathcal{N}(\mathbf{\Phi})$. Any linear combination of vectors in $\mathcal{N}(\PhiB)$ are in $\mathcal{N}(\PhiB)$. Let $\mathcal{S}$ represent a subset of indices of the columns of $\mathbf{Z}$ and let $\vz^{\mathcal{S}} \defeq \sum_{j \in \mathcal{S}} \vz^{(\ve_j)}$. 

First, let $\mathcal{S} = \{1, \ldots , N\}$. Because all off-diagonal components in $\ZB$ are nonnegative 
and because $\vz^{\mathcal{S}}$ must have one negative value, one of the rows of $\mathbf{Z}$ must be entirely zero except for the $-1$ on the diagonal. Note the ordering of the columns in $\ZB$ is arbitrary, so without loss of generality, let this be the first row. Now, let's say that $\mathcal{S} = \{2, 3, \ldots, N\}$. The same logic holds: at least one row must contain all zeros except for the $-1$. Without loss of generality, we can set $[\mathbf{Z}]_{2,3}, [\mathbf{Z}]_{2,4}, \ldots [\mathbf{Z}]_{2,N} = 0$. Repeating this process, we get a lower triangular matrix:
\begin{equation} \label{Z_LT} 
    \mathbf{Z} = 
    \begin{bmatrix}
    -1 & 0 & 0 & \ldots & 0 \\ 
    z^{(\ve_1)}_2 & -1 & 0 & \ldots & 0 \\
    z^{(\ve_1)}_3 & z^{(\ve_2)}_3 & -1 & \ldots & 0 \\
    \vdots & \vdots & \vdots & \ddots & \vdots \\
    z^{(\ve_1)}_N & z^{(\ve_2)}_N & z^{(\ve_3)}_N & \ldots & -1
    \end{bmatrix}.
\end{equation}
However, examining the final column, we see that 
$\vz^{(\ve_N)}$
is a vector of all zeros and one $-1$,
such that it cannot be in $\mathcal{N}({\PhiB})$,
proving Lemma~\ref{lem:exist_onehot} by contradiction.
\end{proof}

\begin{proof}[\unskip\nopunct]

\textit{Proof of Theorem~\ref{th:id_weights}}: Lemma~\ref{lem:exist_onehot} confirms $P<N$, meaning that we can form the concatenated matrix of
$\vz^{\ve_j}$
vectors:
\begin{equation}  \label{Z_P}
    \mathbf{Z} = \begin{bmatrix}
    -1 & 0 & \ldots & 0 \\ 
    z^{(\ve_1)}_2 & -1 & \ldots & 0 \\
    \vdots & \vdots & \ddots & \vdots \\
    z^{(\ve_1)}_P & z^{(\ve_2)}_P & \ldots & -1 \\
    \vdots & \vdots & \ddots & \vdots \\
    z^{(\ve_1)}_N & z^{(\ve_2)}_N & \ldots & z^{(\ve_P)}_N
    \end{bmatrix}.
\end{equation}
Let us now apply 
$P(\mathcal{C}_{\ve_P}|\vlambda) = P(\mathcal{C}_{\ve_P}|\vlambda')$. For all $\vx \in \setC_{\ve_P}$,
\begin{gather}
    \prod_{n=1}^{N} \frac{\lambda_n^{x_n}}{x_n!} - \prod_{n=1}^{N} \frac{{\lambda_n'}^{x_n}}{x_n!} = 0, \\ 
    \bigg(\prod_{\forall i > P}  \frac{\lambda_i ^{x_i}}{x_i!}\bigg)  \bigg( \frac{\lambda_P ^{x_P}}{x_P!} - \frac{{\lambda_P'}^{x_P}}{x_P!} \bigg) = 0, \label{eq:lam_P_diff}
\end{gather}
where Lemma~\ref{lem:sum_lambda} ($\sum_n \lambda_n = \sum_n \lambda_n'$) yields the first equality, and Lemma~\ref{lem:onehot} (all $\lambda_i = \lambda'_i \ \forall i > P$) yields the second equality when combined with the fact that $x_i = 0$ for $i<P$ due to~\eqref{Z_P}.
The only $\vx \in \mathcal{C}_{\ve_P}$ with $x_P \neq 0$ is 
$\ve_P$ which simplifies~\eqref{eq:lam_P_diff} to $\lambda_P = \lambda_P'$.

Now we have $\lambda_i = \lambda'_i \ \forall i > P-1$. Following the same arguments, we can start from 
$P(\mathcal{C}_{\ve_{P-1}}|\vlambda) = P(\mathcal{C}_{\ve_{P-1}}|\vlambda')$ and arrive at $\lambda_{P-1} = \lambda'_{P-1}$.
Applying this repeatedly ultimately yields $\vlambda = \vlambda'$, proving Theorem~\ref{th:id_weights}. \end{proof}

\subsection{Proof of Proposition~\ref{prop:small-single}}

\begin{proof}
By straightforward integration,
\begin{align}
    \label{eq:kappa-g}
    \mathbb{E}_{\vy | n^*} [p(\vy | n)]
    \propto \underbrace{\exp \left\{-\frac{1}{4 \sigma^2} \|\vphi_n - \vphi_{n^*}\|_2^2 \right\}}_{\defeq \kappa(n, n^*)}.
\end{align}
Therefore, given the constraints $\|\vlam\|_1 \leq 1$ and $\lambda_n \geq 0$, the first-order KKT condition is
\begin{align}
    \mathbb{E}_{n^*} \left[ \frac{\vkappa(n^*)}{\langle \vkappa(n^*), \vlambda \rangle} \right] = c \vs - \vmu,
\end{align}
where $\vkappa(n^*) = (\kappa(n, n^*))_{n=0}^N$, $\vs \in \partial \|\vlambda\|_1$, 
$c \geq 0$,
and $\mu_n \geq 0$. By complementary slackness, $\mu_n \lambda_n = 0$ for all $n$. Because $\mK$ is symmetric, we can rewrite the above as
\begin{align}
    \mK \left( \frac{\vlambda^*}{\mK \vlambda} \right) = c\vs - \vmu,
\end{align}
where the fraction represents element-wise division. Solving for $\vlambda$ and rescaling $\vmu$, we obtain the desired expression.
\end{proof}

\subsection{Proof of Theorem~\ref{thm:small-groups}}

\begin{proof}
Let $\kappa_g$ be defined the same as $\kappa$ from \eqref{eq:kappa-g} with $\vphi_n^{(g)}$. Then again by straightforward integration,
\begin{align}
    \tilde\kappa(n, n^*) &\defeq 
    \mathbb{E}_g \left[\kappa_g(n, n^*)\right]\\
    &= \begin{cases}
    1 & n = n^*,\\
    \left( \frac{2 \sigma^2}{2 \sigma^2 + 1} \right)^{M/2} & n \neq n^*, 0 \in \{n, n^*\}, \\
    \left( \frac{\sigma^2}{\sigma^2 + 1} \right)^{M/2} & n \neq n^*, 0 \notin \{n, n^*\}. \\
    \end{cases}
\end{align}
Let $\widetilde \mK = (\tilde \kappa(n, n'))_{n, n'=0}^N$, and let $\widehat \mK = (\tilde \kappa(n, n'))_{n, n'=1}^N$ be the sub-matrix of $\widetilde \mK$ excluding $n=0$. Then
\begin{align}
    \widehat{\mK} = (1 - a) \mI + a \mJ,
\end{align}
where $a = \left( \frac{\sigma^2}{\sigma^2 + 1} \right)^{M/2}$ and $\mJ$ is a matrix of all ones.
Leveraging the block matrix inverse,
we observe that we have the form
\begin{align}
    \widetilde{\mK}^{-1} = \begin{bmatrix}
    \frac{1}{1 - a} \mI + b \mJ & c \vone\\
    c \vone^T & d \\
    \end{bmatrix},
\end{align}
assuming the final column corresponds to $n = 0$, for some scalars $b$, $c$, and $d$. Using the formula from Proposition~\ref{prop:small-single} and the fact that $\vs - \vmu = \vone$ by assumption, we conclude that for $n > 0$, $\widetilde{\lambda}_n \propto \lambda_n^* + C$ for some $C$. Rote algebra verifies that the constant of proportionality is non-negative.
\end{proof}

\section*{Acknowledgments}

This work was supported by NSF grants CBET 2017712,  CCF-1911094, IIS-1838177, and IIS-1730574; ONR grants N00014-18-12571, N00014-20-1-2787, and N00014-20-1-2534; AFOSR grant FA9550-18-1-0478; a Vannevar Bush Faculty Fellowship, ONR grant N00014-18-1-2047; and the Rice University Institute of Biosciences and Bioengineering. P.K.K. was supported by the NLM Training Program in Biomedical Informatics and Data Science (T15LM007093).

\ifCLASSOPTIONcaptionsoff
  \newpage
\fi


\vspace{\ieeebiotrim cm}

\begin{IEEEbiography}[{\includegraphics[width=1in]{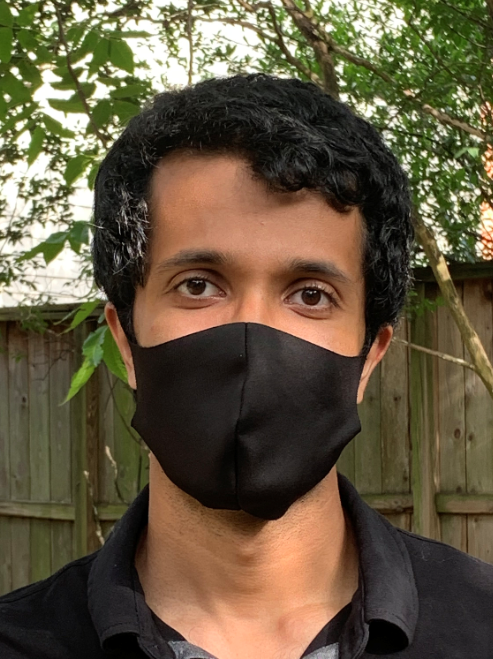}{}}]{Pavan Kota}
received the B.S.E.\ in biomedical engineering from Case Western Reserve University (Cleveland, OH) in 2017.

He is currently pursuing his Ph.D. in the Department of Bioengineering at Rice University and an NIH NLM Fellow in Biomedical Informatics and Data Science. Pavan is primarily interested in the application of signal processing and machine learning techniques to biomedical diagnostics. 
\end{IEEEbiography}

\vspace{\ieeebiotrim cm}

\begin{IEEEbiography}[{\includegraphics[width=1in]{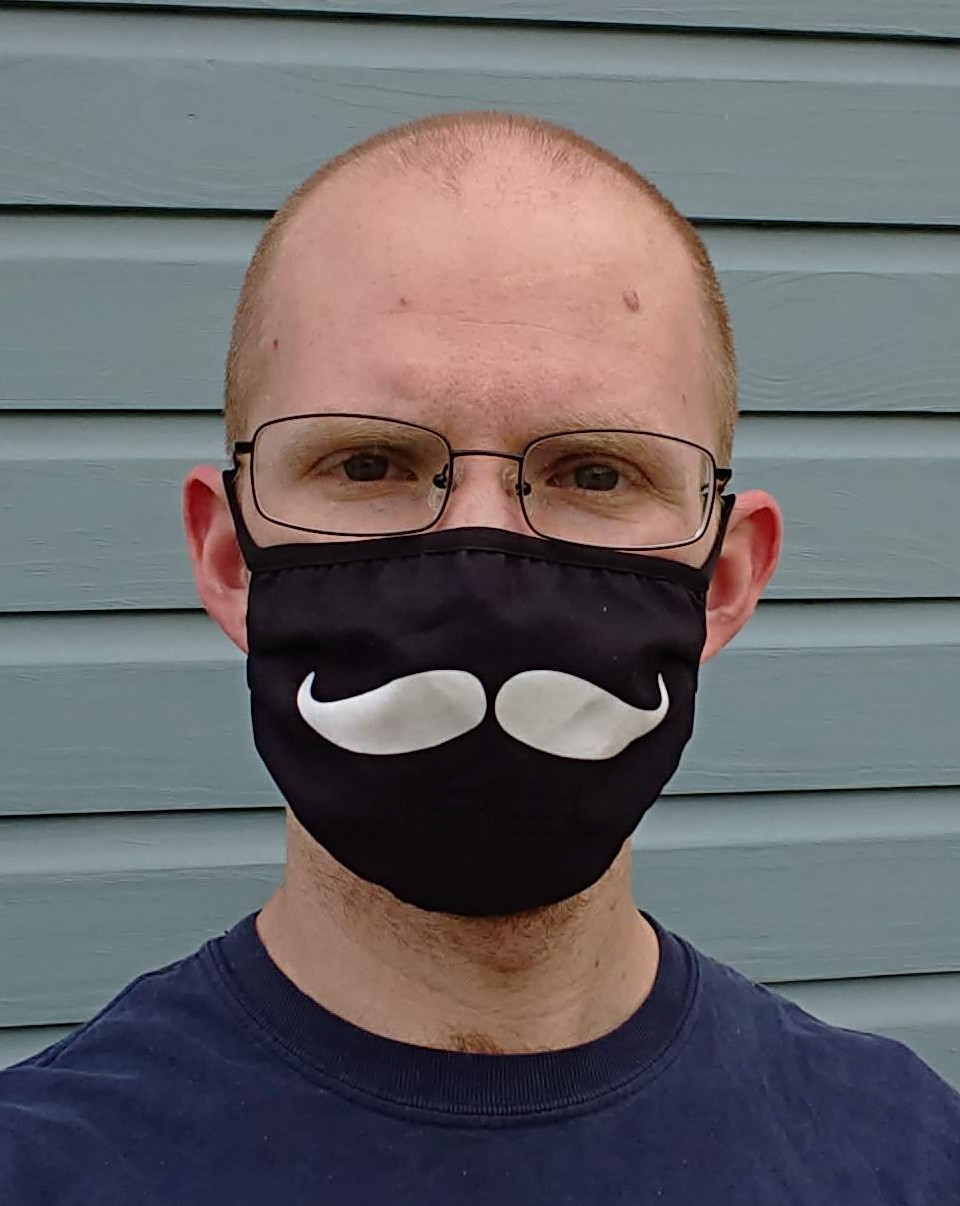}{}}]{Daniel LeJeune}
(S'11)
received the B.S.\ in engineering from McNeese State University (Lake Charles, LA) in 2014 and the M.S.\ in electrical and computer engineering from the University of Michigan in 2016.

He is currently a PhD candidate in the Department of Electrical and Computer Engineering at Rice University. His research interests include machine learning theory and adaptive algorithms.
\end{IEEEbiography}

\vspace{\ieeebiotrim cm}

\begin{IEEEbiography}[{\includegraphics[width=1in]{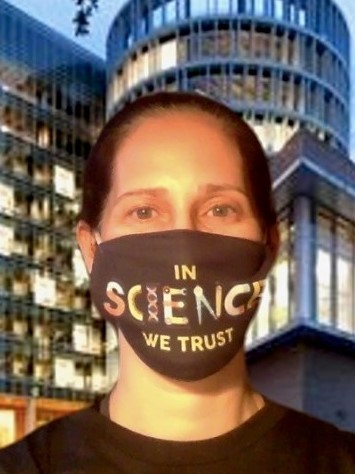}{}}]{Rebekah Drezek}
received her B.S.E. from Duke University (1996) and her M.Sc. (1998) and Ph.D. (2001) from the University of Texas at Austin, all in electrical engineering.

She is currently a professor and associate chair in the Department of Bioengineering at Rice University and a fellow of AIMBE. Her research interests include optical imaging, nanomedicine, and diagnostics.

\end{IEEEbiography}

\vspace{\ieeebiotrim cm}

\begin{IEEEbiography}[{\includegraphics[width=1in]{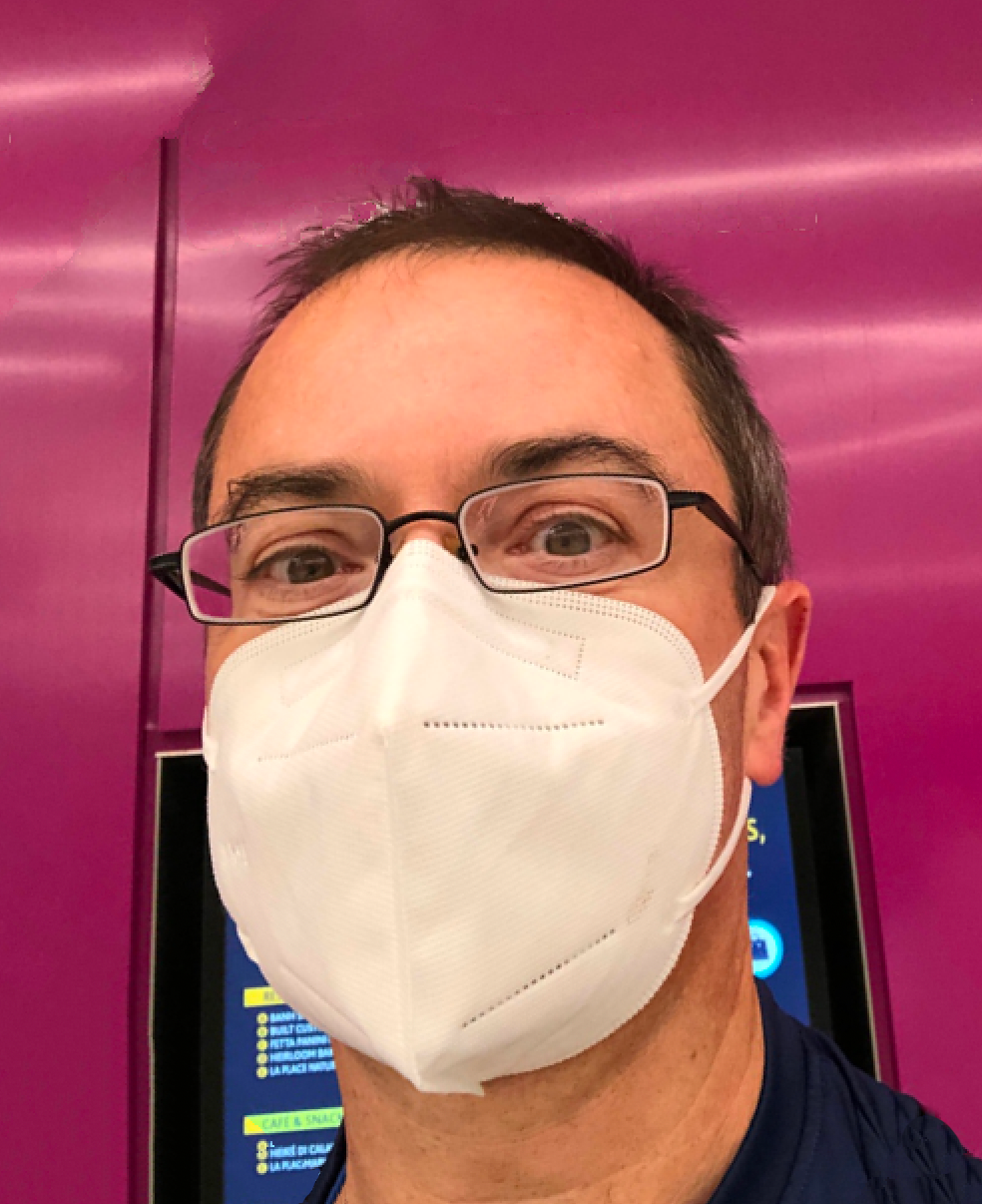}{}}]{Richard Baraniuk}
(S’85--M’93--SM’98--F’01)
received the B.S.\ 
from the University of Manitoba, Canada (1987), the M.S.\ from the
University of Wisconsin-Madison (1988), and the Ph.D.\ from the University of Illinois at Urbana-Champaign (1992), all in electrical engineering.

He is currently the Victor E.\ Cameron Professor of Electrical and Computer Engineering at Rice University and the Founding Director of OpenStax (openstax.org). His research interests lie in new theory, algorithms, and hardware for sensing, signal processing, and machine learning. He is a Fellow of the American Academy of Arts and Sciences, National Academy of Inventors, American Association for the Advancement of Science, and IEEE. He has received the DOD Vannevar Bush Faculty Fellow Award (National Security Science and Engineering Faculty Fellow), the IEEE James H. Mulligan, Jr. Education Medal, and the IEEE Signal Processing Society Technical Achievement, Education, Best Paper, Best Magazine Paper, and Best Column Awards. He holds 35 US and 6 foreign patents.

\end{IEEEbiography}

\end{document}